\documentclass[a4paper,10pt]{amsart}

\usepackage[margin=2cm]{geometry}

\usepackage{enumerate}

\usepackage[T1]{fontenc}
\usepackage[utf8]{inputenc}

\usepackage{amssymb}
\usepackage{amsmath}
\usepackage{epsfig}
\usepackage{subfigure}
\usepackage{graphicx}
\usepackage{breqn} 
\usepackage{multirow}
\usepackage{enumitem}
\usepackage{algorithm}
\usepackage{algorithmic}
\usepackage[usenames]{color}
\usepackage{url}

\usepackage{amsthm}
\newtheorem{theorem}{Theorem}

\theoremstyle{remark}
\newtheorem{example}[theorem]{Example}

\newtheorem*{thm}{\textbf{\textrm{Theorem}}}

\usepackage{tikz}

\usepackage[figuresright]{rotating}

	\newcommand{\EQR}{\operatorname{EQR}}
	\newcommand{\BCH}{\operatorname{BCH}}

\newcommand{\rulesep}{\unskip\ \vrule\ }

\begin{document}

\title[Genetic algorithms with permutation-based representation]{Genetic algorithms with permutation-based representation for computing the distance of linear codes.}

\author{M. P. Cu\'{e}llar}
\email{manupc@decsai.ugr.es}
\address{Department of Computer Science and Artificial Intelligence and CITIC University of Granada, Spain}

\author{J. G\'{o}mez-Torrecillas}
\email{gomezj@ugr.es}
\address{Department of Algebra and IEMath-GR, University of Granada}

\author{F. J. Lobillo}
\email{jlobillo@ugr.es}
\address{Department of Algebra and CITIC, University of Granada}

\author{G. Navarro}
\email{gnavarro@ugr.es}
\address{Department of Computer Science and Artificial Intelligence and CITIC University of Granada, Spain}


\begin{abstract}
Finding the minimum distance of linear codes is an NP-hard problem. Traditionally, this computation has been addressed by means of the design of algorithms that find, by a clever exhaustive search, a linear combination of some generating matrix rows that provides a codeword with minimum weight. Therefore, as the dimension of the code or the size of the underlying finite field increase, so it does exponentially the run time. In this work, we prove that, given a generating matrix, there exists a column permutation which leads to a reduced row echelon form containing a row whose weight is the code distance. This result enables the use of permutations as representation scheme, in contrast to the usual discrete representation, which makes the search of the optimum polynomial time dependent from the base field. In particular, we have implemented genetic and CHC algorithms using this representation as a proof of concept.  Experimental results have been carried out employing codes over fields with two and eight elements, which suggests that evolutionary algorithms with our proposed permutation encoding are competitive with regard to existing methods in the literature. As a by-product, we have found and amended some inaccuracies in the \textsc{Magma} Computational Algebra System concerning the stored distances of some linear codes.
\end{abstract}



\maketitle 

\section{Introduction}\label{Introduction}

Coding Theory is one of the foundational pillars of Information Theory. In his seminal paper \cite{Shannon1948}, Shannon looks for the fundamental limits on signal processing and communication issues. In a noisy channel, redundancy is added to a message in order to recover it whenever transmission errors occur. Shannon's paper relates the correction capability with the error probability of the channel. Given a finite alphabet \(\mathcal{A}\), a block code of length $n$ is a subset \(\mathcal{C} \subseteq \mathcal{A}^n\). The encoding process assigns to each message a codeword, i.e. an element in \(\mathcal{C}\), the codeword is transmitted via a noisy channel and the receiver infers the transmitted codeword from the received word. There are a lot of techniques that can be considered to achieve efficiently the latter task. Historically, a successful strategy consists in taking a finite field \(\mathbb{F}\) as an alphabet, and endowing the codes with some algebraic structure. In this scenario, a linear (block) code is therefore a vector subspace of \(\mathbb{F}^n\). Examples of linear codes are the well known Hamming, Reed-Muller, Reed-Solomon, Goppa, AG or LDPC codes, which have been successfully implemented in CD players, deep space transmissions, wireless communications and more \cite{Huffmann/Pless:2003}. 

The correction capability of a code depends on some distance measure, being the Hamming distance the most usual one. Thus, the computation of the minimum distance of a linear code, or equivalently the minimum weight of all non-zero codewords, becomes a primary task to achieve the error correction capability of the code. However, finding a word with minimum weight for a code is not an easy task. Vardy showed in \cite{Vardy1997} that the decision problem associated to the computation of the minimum distance of a binary linear code is NP-complete. Hence, unless P = NP, we cannot expect to find a general polynomial time exact algorithm to compute the distance of an arbitrary linear code. 

Cryptography, seeking for cryptosystems resistent to attacks from quantum computers, has taken advantages of the latter result. Let us briefly recall the role of linear codes in this setting. Shor's algorithm \cite{Shor94}, that factorizes integers in a quantum computer, is often considered as a major theoretical threat for the current standards in cryptography. As stated in \cite[Page 6]{Nis16}, it is estimated that a quantum computer capable of breaking 2000-bit RSA in a matter of hours could be built by 2030, yielding a serious long-term threat to the cryptosystems currently standardized by the National Institute of Standards and Technology. In this context, Code-based Cryptography \cite{PQCrypto} is one of main research lines looking for a possible new standard concerning Post-quantum Cryptography, i.e. the branch of Cryptography dealing with potentially secure cryptosystems under attacks from quantum computers. Code-based cryptosystems are based on the above-commented difficulty of finding a codeword of minimum weight of an arbitrary linear code, since it is equivalent to decoding a general linear code.  Concretely, the secret key is a code of a family with an efficient decoding algorithm, that is masked so that it looks like a random linear code. Hence any approximate algorithm which is able to compute minimal weight codewords may lead to a potential attack to those cryptosystems.

We can find different approaches in the literature to overcome the distance calculation. One is to design the linear code $\mathcal{C}$ subject to certain constraints using higher algebraic structures, to guarantee a lower bound for the distance $d(\mathcal{C})$. This is the case of a well-known subclass of linear codes named cyclic codes, e.g. BCH \cite{Berlekamp2015} or Reed-Solomon codes \cite{Wicker1994}. These approaches do not tackle the problem of finding the true distance, which remains unknown, although they ensure an error correction capability given by its precalculated lower bound. In many cases, this is enough for practical applications regarding error detection and correction, even if the (potential) full effectiveness of the designed code is not achieved.

Another kind of approaches focuses on finding the true distance using exact algorithms. Here, the problem is considered as a search procedure over a solution space, using heuristics to guide the search and to reduce the computational time of the algorithms. The most known algorithm was designed by Brouwer and Zimmermann \cite{Wassermann/etal:2006}, and later extended by Lisonek and Trummer \cite{LisonekTrummer2016}, and it essays to build information sets from a generating matrix. More recent attempts can be found in the literature, as \cite{joundan19} which focuses on special features of BCH codes and uses a strategy similar to \textit{divide and conquer} to split the code into subcodes. Then, it applies the Brouwer-Zimmermann method to explore smaller subcodes to provide the final code distance.  However, the main drawback is that the efficiency of these methods is still non-polynomial, since the whole solution space has to be explored in the worst case. They have been successfully applied over small binary codes, but their computational time increases exponentially with the length of the code and the bit-size of the underlying finite field. Therefore, for large codes, their time complexity makes its application infeasible. 

A third class of attemtps focuses on providing approximate methods to find lower and/or upper bounds of the distance. In this category, one of the first algorithms was provided by J. S. Leon in \cite{Leon1988}. This method is based on a random sampling of the codewords to approximate the probability that the distance is in a range $[L_B, U_B]$, where $L_B$ and $U_B$ are the lower and upper bounds. Metaheuristics have also been used to solve the problem. Here, the methodology consists on finding  a message word $m$ such that its corresponding codeword $c=mG$ has minimal weight, where $G$ is a generating matrix of the code $\mathcal{C}$. Therefore, the problem formulation uses discrete (often binary) representation to encode a solution $m$ containing $k$ values in the alphabet \(\mathcal{A}\), where $k$ is the dimension of $\mathcal{C}$. Formally, these aim to find an optimal solution $m^*\in \mathcal{A}^k$ that minimizes the fitness $f(m^*)= min_m\{w(mG)\}$ subject to the restriction $f(m^*)>0$ (or equivalently, $m^*\not = 0$), where $w$ denotes the Hamming weight. These methods provide an upper bound of $d(\mathcal{C})$. The first proposals we have found in the literature to tackle the problem  were Genetic Algorithms \cite{DontasDeJong1990} and Simulated Annealing \cite{muxiang_sa_94}. Later, the problem was formulated as a graph traverse problem and a solution with Ant Colony Optimization (ACO) was provided in \cite{BLAND2007391},  using a Tabu Search procedure as local search intensification. The ACO approach was also revisited recently in \cite{8360246}, improving the heuristic to build the ant path. Other metaheuristics, such as Tabu Search, Iterated Local Search, Simulated Annealing, Hill Climbing, and more were compared in \cite{10.1007/978-3-540-73055-2_51}, providing small advances to the problem solution. In these cases, the codes used for experimentation are binary, and the sizes are relatively small (such as [16,7,3]-linear codes in \cite{DontasDeJong1990}, or BCH codes of length $n=255$ and QR codes of length up to $n=199$ in \cite{BLAND2007391} and \cite{8360246}), although the reported results are promising.  An improved Simulated Annealing method to address the problem was recently developed in \cite{BIA0002332} for larger binary codes, such as QR(223, 112, 13), or BCH(511, 286, 55). Other interesting approaches attempt to design the code subject to an established lower bound of the distance by means of Genetic Algorithms, as in \cite{BarbieriCagnoniColavolpe2005}. Regarding specific linear codes, Genetic algorithms have also been employed to exploit properties of BCH and EQR codes to improve the search, as in \cite{Nouhetal13}. 

One of the main limitations in many of the approximate methods is the high selective pressure over the non-valid trivial message $s=(0,0,\ldots ,0)$, which leads to the fitness $f(s)=0$. This problem has been highlighted and studied by some authors, as for instance in \cite{7381308}, and more recently in \cite{Berkani2017ImprovedDO}. In the latter, the effects of the parent selection operator in Genetic Algorithms is studied as an alternative to palliate the high selective pressure problem, although the experimental benefits of the approach are limited to small codes such as Golay(23,12), o BCH(127, 64, 21). One of the most sounded approaches to address the problem of high selective pressure belongs to Askali et al. \cite{askali13}. They design a genetic algorithm to perform the search, and the Multiple Impulse Method (MIM), which is an heuristic that raises the possibility of making random perturbations around nearest zero codewords to find linear combinations that yields an upper bound. The algorithm is tested over a high variety of QR, BCH and double cyclic codes and it is able to obtain the exact distance up to length \(512\) in a relatively low computational time. Finally, in \cite{joundan18} a random local search with the Brouwer-Zimmermann method is used to attack binary BCH and QDC. 

We see that the problem of finding the distance of a linear code is quite complicated and, although it has been tackled by many different perspectives, it remains unsolved for general purposes. 
More specifically, if we focus on metaheuristics, the background problem statement remains the same in all the proposals, and the different authors tackle the problem using different search algorithms, or providing solutions to improve the balance between exploration and intensification. As our interest in this manuscript is to cope with general linear codes, we highlight the following features of the aforementioned ways that will help us to show the contribution of this work: 1) Almost all studied methods focus on the case of binary codes; 2) the solution representation scheme in metaheuristics is therefore the binary encoding, or discrete encoding at most; and 3) the solution representation depends on the size of the alphabet, which is usually the field $\mathbb{F}_2$ with two elements as it has been mentioned. However, codes over larger fields also appear in many applications, see e.g. \cite{Tomas/Rosenthal/Smarandache:2010} for video streaming. 

In this work we propose to address this problem from a completely new perspective. We prove that, given a generating matrix, there exists a column permutation which leads to a reduced row echelon form containing a row whose weight is the code distance. This result allows us to provide a novel problem statement and, therefore, new search mechanisms could be applied to provide a problem solution. In our case, we propose to use a permutation representation of the solution space to search the minimum distance. Hence a Generational Genetic Algorithm (GGA) and a Cross generational elitist selection, Heterogeneous recombination, and Cataclysmic mutation algorithm (CHC) are designed for solving the problem. Our methods are tested with codes of medium size over the fields with two and eight elements. We then compare the experimental results with the ones obtained by some of the algorithms described above. It is worth mentioning that, in the experimental section, we apply the proposal over finite fields larger than $\mathbb{F}_2$, which have not been tackled before because of the intrinsic complexity of the classic solution representation. Therefore, as baseline methods,  we have also implemented GGA and CHC algorithms using the traditional discrete representation. Unlike the state-of-the-art search algorithms, the complexity of our proposal grows polynomially with respect to the bit-size of the elements of the base field. This has allowed us to discover some inaccuracies in the \textsc{Magma} Computational Algebra System concerning the stored distances of some linear codes. Our experiments suggest that this change in the solution representation allows us to overcome local optima solutions and speed up the search of the minimum distance with respect to the traditional approaches.

The article is structured as follows: Section \ref{sec2} contains background concepts on mathematics and coding theory to make this article self-contained and introduce the notation. After that, Section \ref{sec:permRep} proposes the permutation representation as a solution encoding for the problem of finding the minimum distance of linear codes. Section \ref{sec:algs} describes both the baseline and proposed algorithms used to solve the problem, and Section \ref{Section3} shows the experiments performed. Section \ref{sec:magma} extends the experimentation for specific cases for which we have found new upper bounds for the code distance, and Section \ref{Section4} contains the conclusions.

\section{Background}\label{sec2}

The mathematical background of this paper relies in several aspects of finite fields, linear algebra and coding theory, which are revised in this section to fix the notation and for self-completeness.

\subsection{Finite fields}\label{finitefields}

Finite fields, also called Galois fields \cite{Lidl/Niederreiter:1996}, are natural generalizations of prime fields \(\mathbb{F}_p = \mathbb{Z}/(p)\), the ring of integers modulo a prime number \(p\). The elements of a finite field \(\mathbb{F}_q\) of \(q = p^r\) elements can be seen as polynomials over \(\mathbb{F}_p\) of degree less than \(r\), and the arithmetic is performed modulo an irreducible polynomial of degree \(r\). For example, the field \(\mathbb{F}_8\) with eight elements can be presented as polynomials in \(\mathbb{F}_2[a]\) of degree less than \(3\), and in this paper the field operations are done modulo the polynomial \(a^3 + a + 1\). The sum is the usual sum of polynomials, whilst the product is obtained by performing the usual product over polynomials modulo \(a^3 + a + 1\). For example, the product of $(a^2 + 1)$ and $(a^2 + a + 1)$ in \(\mathbb{F}_8\) is calculated as follows:
	\[
	(a^2 + 1)(a^2 + a + 1) = a^4 + a^3 + a + 1 \bmod{a^3 + a + 1} = a^2 + a,
	\]
	or in bit notation
	\[
	101 \cdot 111 = 110.
	\]
Non zero elements of a finite field form a cyclic group, i.e. they can be represented as powers of an element (called primitive). Primitive elements allow a more efficient multiplication using a tabular data structure. The following table shows that \(a \in \mathbb{F}_8\) is a primitive element and all elements of $\mathbb{F}_8$ but the zero element can be calculated as powers of $a$. 
	\[
	\begin{array}{|r|r|r|}
	\hline
	\text{Power} & \text{Polynomial} & \text{Bits} \\
	\hline
	a^0 & 1 & 001 \\
	a^1 & a & 010 \\
	a^2 & a^2 & 100 \\
	a^3 & a+1 & 011 \\
	a^4 & a^2 + a & 110 \\
	a^5 & a^2 + a + 1 & 111 \\
	a^6 & a^2 + 1 & 101 \\
	a^7 & 1 & 001 \\
	\hline
	\end{array}
	\]
	The product can be computed using the corresponding primitive representation and that \(a^7 = 1\). For instance,
	\[
	(a^2 + 1)(a^2 + a + 1) = a^6 a^5 = a^{11 \bmod{7}} = a^4 = a^2 + a.
	\]

\subsection{Reduced row echelon form}\label{linearalgebra}

Let \(\mathbb{F}\) be a field and let \(\mathbb{F}^{k \times n}\) denote the set of \(k \times n\) matrices over \(\mathbb{F}\). Recall that a matrix \(H \in \mathbb{F}^{k \times n}\) is in reduced row echelon form (RREF for short), see \cite[Chapter 5]{Liesen/Mehrmann:2015}, if all zero rows are below the non zero rows, the pivot (i.e. the first non zero element) in a non zero row is \(1\), the pivot of a non zero row is strictly to the right of the pivot of the row above, and each column containing a pivot has zeros everywhere else. The notion of echelon form can be found in almost every Linear Algebra book. We suggest \cite[Chapter 5]{Liesen/Mehrmann:2015} for additional reading. For completeness we recall the following basic properties about the RREF because they are essential for the mathematical basis of our proposal: 
	\begin{enumerate}
	\item Two matrices \(M,N \in \mathbb{F}^{k \times n}\) are said to be row equivalent if there exists a non singular matrix \(S \in \mathbb{F}^{k \times k}\) such that \(M = SN\), i.e. each row of $M$ is a linear combination of rows of $N$, and viceversa.
		\item Two matrices \(M,N \in \mathbb{F}^{k \times n}\) are row equivalent if and only if their rows generate the same vector subspace of \(\mathbb{F}^n\).
		\item Each \(M \in \mathbb{F}^{k \times n}\) is row equivalent to a matrix \(H \in \mathbb{F}^{k \times n}\) in RREF. 
		\item If \(H,H' \in \mathbb{F}^{k \times n}\) are in RREF and row equivalent, then \(H = H'\).
		\item If \(v \in \mathbb{F}^n \setminus \{0\}\) belongs to the vector subspace generated by the rows of \(M \in \mathbb{F}^{k \times n}\), then \(M\) is row equivalent to a matrix \(\left(\begin{array}{c} M' \\ \hline v \end{array} \right)\), where \(M' \in \mathbb{F}^{(k-1) \times n}\). 
	\end{enumerate}

As a consequence, any matrix is row equivalent to a unique matrix in RREF. Moreover, the non zero rows of the RREF are a basis for this subspace.

\subsection{Linear error correcting codes}\label{codingtheory}

In the most general setting, an error correcting code is simply a set $\mathcal{C}$,  whose elements are called codewords, aiming to increase the reliability of the communication through a noisy channel. Given a set $\mathcal{M}$ of messages, each message $m\in \mathcal{M}$ is one-to-one assigned to a codeword $c$ which is transmitted. If $c$ is corrupted by the noise, the receiver could get an element $y\notin \mathcal{C}$, detecting an error in the communication. If an eventual decoding algorithm was available, the receiver would correct $y$ to the codeword $c$ and recover the original message $m$. 

In practice, as commented in the Introduction, a rich algebraic structure on a code usually leads to a better knowledge of its basic parameters as well as the design of faster encoding and decoding algorithms. Let \(\mathbb{F}_q\) be the field with \(q\) elements. An \([n,k]_q\)-linear code \(\mathcal{C}\) is a \(k\)-dimensional vector subspace of \(\mathbb{F}_q^n\).  Linear codes can be provided by generating matrices or by parity check matrices. Concretely, a generating matrix for \(\mathcal{C}\) is any $k\times n$ matrix \(G \in \mathbb{F}_q^{k \times n}\) such that 
\[\mathcal{C} = \left\{mG ~;~ m \in \mathbb{F}_q^{k}\right\},\] 
i.e. each codeword can be obtained in a unique way multiplying a vector message to be encoded \(m \in \mathbb{F}^k\) by the matrix \(G\). Since \(\mathcal{C}\) has dimension \(k\), it follows that the rows of \(G\) form a basis of the code.  For clarity, we sometimes remark the dimensions of the matrix by using subindices, so that the matrix $G$ is denoted as $G_{k,n}$. A parity check matrix is any matrix \(H \in \mathbb{F}_q^{n \times (n-k)}\) such that
\[
\mathcal{C} = \left\{v \in \mathbb{F}_q^n ~;~ v H = 0\right\}.
\]
Parity check matrices are quite useful since they are linked to a general decoding strategy called syndrome decoding.

One of the most important parameters of a code is its distance. Let $x\in \mathbb{F}_q^n$, we denote by \(\operatorname{w}(x)\) the number of non zero entries in \(x\), i.e. the Hamming weight of \(x\). This is a weight function which has a corresponding associated distance, \(\operatorname{d}(x,y) = \operatorname{w}(x-y)\), measuring the number of different coordinates of \(x\) and \(y\). The (minimum Hamming) distance $\operatorname{d}(C)$ of a linear code \(\mathcal{C}\) is defined as
\[
\operatorname{d}(C) = \min\left\{\operatorname{w}(c)~;~c\in\mathcal{C}\setminus\{0\}\right\},
\]
or equivalently, as 
\[
\operatorname{d}(C) = \min\left\{\operatorname{d}(c,c')~;~c,c'\in\mathcal{C}, c \neq c' \right\}.
\]
Then, if a lower bound \(d\) of the distance is known, we refer to \(\mathcal{C}\) as an \([n,k,d]_q\)-linear code. 

\begin{example}
Let \(\mathcal{C}\) be the linear code over the field with eight elements $\mathbb{F}_8$ defined by the generating matrix
\[
G_{3,6} = \begin{pmatrix}
1 & 0 & 0 & a^5 & a^{2} & a^4 \\
0 & 1 & 0 & a & a^5 & 1 \\
0 & 0 & 1 & 0 & a^5 & a^5
\end{pmatrix}.
\]
The rows of $G$ form a basis of $\mathcal{C}$, so that  $\mathcal{C}$ is a $[3,6]_8$-linear code. In this setting, the message $m=(1,a,1)\in \mathbb{F}^3_8$ is encoded to the codeword
$$m\,G_{3,6}=(1,a,1) \begin{pmatrix}
1 & 0 & 0 & a^5 & a^{2} & a^4 \\
0 & 1 & 0 & a & a^5 & 1 \\
0 & 0 & 1 & 0 & a^5 & a^5
\end{pmatrix}=(1,a,1,a^3,a^4,a^3)\in \mathbb{F}^6_3.$$
The reader may check that 
$$(1,a^4,0,0,0,0)= 1\cdot (1,0,0,a^5,a^2,a^4)+a^4\cdot (0,1,0,a,a^5,1)$$ is a codeword of minimum weight. Hence \(\mathcal{C}\) is a \([6,3,2]_8\)-linear code. Observe that no row of \(G\), which is in RREF, has Hamming weight \(2\). 
\end{example}

The distance of a code is linked to its correction capability. Concretely, if $d$ is the minimum distance, \(t = \left\lfloor \frac{d-1}{2} \right\rfloor\) is the so called error-correcting capability. The balls centered in the codewords with radius \(t\) are pairwise disjoint.  Thus, for any received word of the form $y = c + e$, where \(c \in \mathcal{C}\) and $e$ is an error  vector of weight less than $t$, \(c\) is the closest codeword to \(y\), see Figure \ref{deccapability}. 
\begin{figure}[ht]
$$\begin{tikzpicture}[scale=1, main node/.style={circle, draw, font = \bfseries}]
\node at (0,0) (0) {$c'$};
\node at (4,1) (0') {$c$};
\node at (0,-0.3) (1) {$\circ$};
\node at (4,0.7) (1') {$\circ$};
\draw[<->, dashed] (1) -- (1');
\draw (0,-0.3) circle (1.6);
\draw (4,0.7) circle (1.6);
\node at (2,0.5) {$d$};
\node at (-1.6,-0.3) (2) {};
\draw[<->,dashed] (2) -- (1);
\node at (5,0) (3) {$\circ$};
\node at (5.1,0.3) (5) {$y$};
\node at (-0.8,0) (4) {$t$};
\draw[->] (1') -- (3);
\node at (4.3,0.2) (6) {$e$};
\end{tikzpicture}$$
\caption{Balls centered at two codewords}\label{deccapability}
\end{figure}
A nearest neighbor decoding algorithm finds the closest codeword of a given word. These results can be checked in \cite[\S 1.11]{Huffmann/Pless:2003}. 

The latter is an important property regarding our approach in this paper since it makes closer the minimum distance computation and the decoding process. Assume we have a general algorithm \(\operatorname{mincw}()\) for computing the distance of a linear code, that is, given a generating matrix  of a linear code, its output is a non zero codeword of minimum weight. Let then \(\mathcal{C}\) be an \([n,k,d]\)--code with generating matrix \(G\). If $c\in \mathcal{C}$ is transmitted and \(y = c + e\) is received, where $e$ is an error vector of weight lower or equal that \(t = \left\lfloor \frac{d-1}{2} \right\rfloor\), it is easy to see that \(e\) is the only codeword of minimum weight of the linear code \(\mathcal{C} + \langle y \rangle\), where $\langle y \rangle$ denotes the vector space generaed by $y$. Indeed, suppose that \(e' = c' + y\) with \(c' \in \mathcal{C}\) and \(\operatorname{w}(e') \leq \operatorname{w}(e) \leq t\), then \(e - e' = y - c + c' - y = c' - c \in \mathcal{C}\) and \(\operatorname{w}(e - e') \leq \operatorname{w}(e) + \operatorname{w}(e') \leq t + t < d\). Since \(d\) is the minimum distance of \(\mathcal{C}\), necessarily $e=e'$. The block matrix \(G' = \left(\begin{array}{c} G \\ \hline y \end{array}\right)\) is a generator matrix for \(\mathcal{C} + \langle y \rangle\), it follows \(\operatorname{mincw}(G') = e\). Therefore such an algorithm \(\operatorname{mincw}()\) yields a general decoding algorithm up to the correction capability of the code, see Algorithm \ref{gendec}.
\floatname{algorithm}{Algorithm}
\begin{algorithm}[ht]
    \caption{General decoding algorithm}\label{gendec}
    \begin{algorithmic}[1]
    \REQUIRE A received word $y$. \textbf{Requirements:} a minimal weight codeword computation algorithm  \(\operatorname{mincw}()\) and a generating matrix $G$ of the linear code.
    \ENSURE A codeword $c$.
       \STATE $G'\gets \left(\begin{array}{c} G \\ \hline y \end{array}\right)$.
        \STATE $e\gets \operatorname{mincw}(G')$.
        \RETURN  $y-e$.
    \end{algorithmic}
\end{algorithm}

The computation of the minimum distance, and a corresponding non zero codeword of minimal weight, is therefore a major issue for a linear code. As shown in \cite[Theorem 5]{Vardy1997}, finding a non zero codeword of minimal weight of a linear code is an NP-hard problem. Hence, unless P = NP, we cannot expect to design an efficient algorithm to compute the distance of an arbitrary linear code. This result on complexity theory has two consequences. On the one hand, to overcome this limitation, most codes are constructed with additional algebraic structures in order to fix a known lower bound of their distances and correction capability. Examples can be found in \cite{Bose/Chaudhuri:1960,Hocquenghem:1959,Hartmann/Tzeng:1972} for cyclic codes and \cite{Gomez/Lobillo/Navarro/Neri:2018} for skew cyclic codes. However, in this paper we assume linear codes that are not constrained by any additional structure, and therefore our approach can be used for any family of linear codes. On the other hand, the NP-hardness of the general minimum distance computation is the strongest tool in code-based cryptography, a promise branch of the so-called Post-quantum Cryptography, that comprises the study of cryptographic algorithms thought to be secure against an attack by a quantum computer. Hence any algorithm which is able to compute minimal weight codewords may lead to a potential attack to those cryptosystems. We recommend \cite{OverbeckSendrier2009} as a good source in code-based cryptography. The reader may also consult the competition promoted by the NIST (National Institute of Standards and Technology), regarding to standardize one or more quantum-resistant public-key cryptographic algorithms, where Classic McEliece proposal in \url{https://classic.mceliece.org/}, a code-based cryptosystem, is one of the finalists.

\section{Permutation-based scheme}\label{sec:permRep}

Let \(0 < k \leq n\) be two non negative integers, and \(\mathbb{F}_q\) the field with \(q\) elements. We denote by \(\mathcal{S}_n\) the set of permutations of \(n\) symbols.  Two \([n,k]_q\)-linear codes \(\mathcal{C}_1\) y \(\mathcal{C}_2\) are said to be permutation equivalent if they are equal up to a fixed permutation on the codeword coordinates, that is, if there exists a permutation \(x \in \mathcal{S}_n\) such that \((c_0, c_1, \dots, c_{n-1}) \in \mathcal{C}_1\) if and only if \((c_{x(0)}, c_{x(1)}, \dots, c_{x(n-1)}) \in \mathcal{C}_2\). It is easy to see that \(G\) is a generating matrix for \(\mathcal{C}_1\) if and only if \(GP\) is a generating matrix for \(\mathcal{C}_2\), where \(P\) is the permutation matrix associated to \(x\). Although permutation equivalent codes could be different, they share the same minimum distance, since a permutation of its components does not modify the Hamming weight of a vector.

Our new perspective to compute the minimum distance relies in the the following result, whose proof uses the properties of the RREF stated in subsection \ref{linearalgebra}. See also \cite{GLN2019}.

\begin{thm}\label{mathmain}
Let $G$ be a $k \times n$ generating matrix of a $[n,k]_q$-linear code $\mathcal{C}$ over the finite field $\mathbb{F}_q$. There exists a permutation $x \in \mathcal{S}_n$ such that the RREF, $R$, of $G P_x$, where $P_x$ is the permutation matrix of $x$, satisfies that the Hamming weight of some of its rows reaches the minimum distance of $\mathcal{C}$. Consequently, if $b$ is a row of $R$ verifying such property, then $bP_x^{-1}$ is a non zero codeword of $\mathcal{C}$ with minimal weight.
\end{thm}

\begin{proof}
Let \(d\) be the minimum distance of \(\mathcal{C}\) and let \(a \in \mathcal{C}\) be a codeword such that \(\operatorname{w}(a) = d\). Since \(a\) is non zero, there exists a non singular matrix $A \in \mathbb{F}_q^{k \times k}$ such that 
\[
A G = \left(\begin{array}{c} G_1 \\ \hline a \end{array}\right) 
\]
for some \(G_1 \in \mathbb{F}_q^{(k-1) \times n}\). Since \(\operatorname{w}(a) = d\), there exists a permutation \(x \in \mathcal{S}_n\) moving the non zero entries of \(a\) to the last positions, i.e.
\[
a P_x = (0, \dots, 0, a_1, \dots, a_d),
\]
where \(a_1, \dots, a_d \in \mathbb{F}_q \setminus \{0\}\). Hence 
\[
AGP_x = \left(\begin{array}{c} G_1 P_x \\ \hline a P_x \end{array}\right) = \left( \begin{array}{c} G_1 P_x  \\\hline
\begin{array}{c|c} 0 \cdots 0 & a_1 \cdots a_d \end{array} \end{array} \right).
\]
Now, let \(A' \in \mathbb{F}_q^{(k-1) \times (k-1)}\) invertible such that \(A' G_1 P_x = H'\) is in RREF. Hence,
\begin{equation}\label{eq:rref}
\left( \begin{array}{c|c} A' & 0\\ \hline 0 & 1 \end{array} \right) AGP_x = \left( \begin{array}{c|c} A' & 0\\ \hline 0 & 1 \end{array} \right) \left( \begin{array}{c} G_1 P_x  \\\hline
\begin{array}{c|c} 0 \cdots 0 & a_1 \cdots a_d \end{array} \end{array} \right) = \left( \begin{array}{c} H'  \\\hline
\begin{array}{c|c} 0 \cdots 0 & a_1 \cdots a_d \end{array} \end{array} \right)
\end{equation}
Since $G_1$ has rank $k-1$, the last row of $H'$ is nonzero. Suposse that the pivot of this row is in the $i_0$-th column. If $i_0<n-d+1$, then the last row of \eqref{eq:rref} is the last row of the RREF of $GP_x$ up to non zero scalar multiplication, and we are done. Otherwise, the last two rows of \eqref{eq:rref} are linearly independent and their nonzero coordinates are placed at the last $d$ coordinates. Hence, there exists a linear combination of both whose hamming weight is lower than $d$, a contradiction. The last statement is straightforward.
\end{proof}

By Theorem, finding the minimum distance of an \([n,k]_q\)-linear code is reduced to find the minimum of the map $\mathfrak{d}:\mathcal{S}_n\to \mathbb{N}$ defined by
\begin{equation}\label{eq:fitness}
\mathfrak{d}(x)=\min \left\{\operatorname{w}(b) ~|~ b \text{ is a row of the RREF of } GP_x \right\}
\end{equation}
for any $x\in \mathcal{S}_n$, where $G$ is a generating matrix of the code and $P_x$ represents the permutation matrix of  $x$.
This encoding is then invariant with respect to the base field. Obviously, the computation of $\mathfrak{d}(x)$, for some permutation \(x\), does depend on $q$ and $n$. However the row reduction can be performed in \(\mathcal{O}(n^3)\) operations in $\mathbb{F}_q$ using the classical scholar algorithm. In fact, using fast matrix multiplication and the so called striped matrix reduction, the complexity could be reduced to \(\mathcal{O}(n^{\log_27})\) and \(\mathcal{O}(n^3/\log n)\), see \cite{BunchHopcroft1974, AndrenHellstromMarkstrom2007}, respectively.

Anyway, there are just a few permutations which provide the minimum of \(\mathfrak{d}\). Concretely, assume there is only one codeword which achieves the minimum weight \(d\) of \(\mathcal{C}\) up to scalar multiplication. Then the probability of finding the minimum of \(\mathfrak{d}\) by a random search is 
\[
\frac{d! (n-d)!}{n!} = \binom{n}{d}^{-1},
\] 
since, given a permutation \(x \in \mathcal{S}_n\) such that \(a P_x = (0, \dots, 0, a_1, \dots, a_d)\) for a codeword \(a \in \mathcal{C}\) of minimal weight \(\operatorname{w}(a) = d\), any other permutation of the last \(d\) and the first \(n-d\) positions of \(a P_x\) also yields a RREF whose last row reaches the minimum weight. In the general case we obtain a lower bound of the probability of finding the minimum of \(\mathfrak{d}\) by random search. There are codes with just one codeword of minimum weight (up to scalar multiplication), hence this lower bound can be achieved. 

\begin{example} \label{ex:rep}
Let us illustrate the proposed scheme by means of a toy example. Let $\mathcal{C}$ be the $[8,4]_4$-linear code over the field with four elements $\mathbb{F}_4=\{0,1,a,a^2=a+1\}$ with generating matrix
$$G_{4, 8}=\left(\begin{array}{cccccccc}
1 & 0 & a + 1 & a + 1 & a + 1 & 0 & 0 & a \\
a + 1 & 0 & 0 & a + 1 & a & a + 1 & 1 & a + 1 \\
1 & a + 1 & 1 & a + 1 & a & 0 & a & 0 \\
0 & 0 & 0 & a & a & a & a + 1 & a
\end{array}\right).$$
So that our ultimate aim to find the minimum of the above-described map $\mathfrak{d}:\mathcal{S}_8\to \mathbb{N}$. Let $x\in \mathcal{S}_8$ be the permutation given by  $x(0)= 1$, $x(1)=0$, $x(2)=3$, $x(3)=2$, $x(4)=5$, $x(5)=4$, $x(6)=7$ and $x(7)=6$, that is, $x$ is the composition of the cycles $(0,1)$, $(2,3)$, $(4,5)$ and $(6,7)$. We apply the permutation $x$ to the columns of $G_{4, 8}$ in order to obtain a generating matrix $G^x_{4,8}$ of a permutation equivalent code $\mathcal{C}_x$, see Figure \ref{appper}.

\begin{figure}[ht]
$$\begin{tikzpicture}[scale=1, main node/.style={circle, draw, font = \bfseries}]
\node at (0,3.6) {$G_{4, 8}=\left(\begin{array}{cccccccc}
1 & 0 & a + 1 & a + 1 & a + 1 & 0 & 0 & a \\
a + 1 & 0 & 0 & a + 1 & a & a + 1 & 1 & a + 1 \\
1 & a + 1 & 1 & a + 1 & a & 0 & a & 0 \\
0 & 0 & 0 & a & a & a & a + 1 & a
\end{array}\right)$};
\node[main node] at (-3.29,2.4) (0) {\scriptsize 0};
\node[main node] at (-2.21,2.4) (1) {\scriptsize 1};
\node[main node] at (-1.12,2.4) (2) {\scriptsize 2};
\node[main node] at (-0.04,2.4) (3) {\scriptsize 3};
\node[main node] at (1.03,2.4) (4) {\scriptsize 4};
\node[main node] at (2.1,2.4) (5) {\scriptsize 5};
\node[main node] at (3.2,2.4) (6) {\scriptsize 6};
\node[main node] at (4.27,2.4) (7) {\scriptsize 7};

\node[main node] at (-3.29,1) (0') {\scriptsize 0};
\node[main node] at (-2.21,1) (1')  {\scriptsize 1};
\node[main node] at (-1.12,1) (2') {\scriptsize 2};
\node[main node] at (-0.04,1) (3') {\scriptsize 3};
\node[main node] at (1.03,1) (4') {\scriptsize 4};
\node[main node] at (2.1,1) (5') {\scriptsize 5};
\node[main node] at (3.2,1) (6') {\scriptsize 6};
\node[main node] at (4.27,1) (7') {\scriptsize 7};
\node at (0,-0.3) {$G^x_{4, 8}=\left(\begin{array}{cccccccc}
0 & 1 & a + 1 & a + 1 & 0 & a + 1 & a & 0 \\
0 & a + 1 & a + 1 & 0 & a + 1 & a & a + 1 & 1 \\
a + 1 & 1 & a + 1 & 1 & 0 & a & 0 & a \\
0 & 0 & a & 0 & a & a & a & a + 1
\end{array}\right)$};
\node at (-4,1.75) {$x$};
\draw[->] (0) -- (1');
\draw[->] (1) -- (0');
\draw[->] (2) -- (3');
\draw[->] (3) -- (2');
\draw[->] (4) -- (5');
\draw[->] (5) -- (4');
\draw[->] (6) -- (7');
\draw[->] (7) -- (6');
\end{tikzpicture}$$
\caption{Construction of the permuted generating matrix}\label{appper}
\end{figure}

Actually, $G^x_{4, 8}$ could be obtained by multiplying $G_{4, 8}$ by the right by the permutation matrix
$$
P_x=\left(\begin{array}{cccccccc}
0 & 1 & 0 & 0 & 0 & 0 & 0 & 0 \\
1 & 0 & 0 & 0 & 0 & 0 & 0 & 0 \\
0 & 0 & 0 & 1 & 0 & 0 & 0 & 0 \\
0 & 0 & 1 & 0 & 0 & 0 & 0 & 0 \\
0 & 0 & 0 & 0 & 0 & 1 & 0 & 0 \\
0 & 0 & 0 & 0 & 1 & 0 & 0 & 0 \\
0 & 0 & 0 & 0 & 0 & 0 & 0 & 1 \\
0 & 0 & 0 & 0 & 0 & 0 & 1 & 0
\end{array}\right).
$$
Now, we compute the RREF of the matrix $G^x_{4, 8}$, and the Hamming weight of each row, obtaining
$$\left(\begin{array}{cccccccc}
1 & 0 & 0 & 0 & a + 1 & 0 & a & a \\
0 & 1 & 0 & 0 & 0 & a & 0 & 0 \\
0 & 0 & 1 & 0 & 1 & 1 & 1 & a \\
0 & 0 & 0 & 1 & 1 & a + 1 & a & a
\end{array}\right) 
\begin{array}{cc}
\rightarrow & 4 \\
\rightarrow & 2 \\
\rightarrow & 5 \\
\rightarrow & 5
\end{array}
$$
Observe that the second row has minimal Hamming weight among all the rows, thus $\mathfrak{d}(x)=2$. Actually, the distance of $\mathcal{C}$ is 2, so $x$ reaches the minimum of $\mathfrak{d}$.
\end{example}

\section{Algorithms}\label{sec:algs}

The theorem 
 of Section \ref{sec:permRep} has an immediate consequence: It allows us to change the classic problem statement of finding the minimum distance, so that the search of a codeword with minimum weight can be stated as finding a suitable permutation of the columns of a generating matrix. In this section we take advantages of this perspective and design a Generational Genetic Algorithm (GGA) and a Cross generational elitist selection, Heterogeneous recombination, and Cataclysmic mutation algorithm (CHC) for solving the problem. We first include on them the traditional (discrete representation) problem statement as baseline methods for comparison, and afterwards we adapt these algorithms to the new permutation representation scheme. Thus, four algorithms are performed in the experimentation: \textit{GGA-Discrete} and \textit{GGA-Order}, which implement a generational genetic algorithm with traditional discrete representation and with permutation representation, respectively; and \textit{CHC-Discrete} and \textit{CHC-Order} with also discrete and permutation representation, respectively. The aim of considering the CHC algorithms is to study the balance between diversity and convergence in the GGA algorithms, mainly for the traditional approach, to show the problems regarding high selective pressure towards the trivial non-valid zero solution and how they are trapped into local optima. Although CHC is an evolutionary algorithm whose initial version was proposed for binary encoding \cite{eshelman1991}, there have been proposals to adapt the algorithm for real and permutation encoding in \cite{eshelman93, cordon2006, chcdiscrete1}. We rely on these papers to build the problem-specific components of the algorithms. A preliminary version of the GGA-Order algorithm was previously presented as a conference poster \cite{GLN2019} and published as an abstract.

Although both GGA and CHC algorithms are extensively known, Algorithms \ref{alg:GGA} and \ref{alg:chc} give a generic description of their execution steps, respectively, to make the manuscript self-contained and ease the reproducibility of the experiments. We now first describe the common components to both discrete and order representation and we particularize each case in the forthcoming subsections.

\subsection{Common components}

The GGA's follow the procedure described in Algorithm \ref{alg:GGA}. It is the classic genetic algorithm with elitism and reinitialization \cite{BackFogelMichalewicz1997}, and it also contains a reconstruction procedure when a non-valid solution is generated by crossover or mutation (lines 12--13 in Algorithm \ref{alg:GGA}). This reconstruction procedure substitutes the non-valid solutions with valid random solutions to encourage exploration. The non-valid solution in the \textit{GGA-Discrete} algorithm is the zero solution $(0, 0,\ldots , 0)$. The reconstruction procedure is not included in the \textit{GGA-Order} implementation, since all possible solutions generated are permutations, which are valid in this case. 

\floatname{algorithm}{Algorithm}
\begin{algorithm}[ht]
    \caption{Generational Genetic Algorithm (GGA)}\label{alg:GGA}
    \begin{algorithmic}[1]
        \REQUIRE $N = \text{Even natural number with the size of the population}$
        \REQUIRE $p_c = \text{Crossover probability}$
        \REQUIRE $MaxReinit = \text{Number of solution evaluations with no fitness improvement before reinitialization}$
        \STATE Set $t$ = 0
        \STATE Initialize the population $P(t)$ with $N$ random valid solutions
        \STATE Evaluate solutions in $P(t)$
        \WHILE {stopping criterion not satisfied}
            \STATE Set $P(t+1)$ $= \emptyset$
            \STATE Set $parents(1..N)$ = Selection of $N$ solutions in $P(t)$ with binary tournament selection
            \FOR{i in 0..$N/2-1$}
                \IF{random number from uniform distibution in \([0,1]\) is less than $p_c$}
                \STATE Set $c_1, c_2$ = solutions generated from crossover over $parents(2i+1)$ and $parents(2i+2)$
                \ELSE
                \STATE Set $c_1$ = mutation of $parent(2i+1)$, $c_2$ = mutation of $parent(2i+2)$
                \ENDIF
                \IF{$c_1$ (resp. $c_2$) is not valid}
                \STATE replace $c_1$ (resp. $c_2$) with a random valid solution
                \ENDIF
                \STATE Update $P(t+1)= P(t+1) \bigcup \{c_1,c_2\}$
            \ENDFOR
            \STATE Evaluate solutions in $P(t+1)$
            \IF{no solution fitness in $P(t+1)$ is better or equivalent to the best solution fitness in $P(t)$}
            \STATE Replace worst solution in $P(t+1)$ with the best solution in $P(t)$
            \ENDIF
            \IF{$MaxReinit$ solutions were evaluated with no improvement regarding the best solution in $P(t+1)$}
			\STATE Replace solutions in $P(t+1)$ with $N-1$ random solutions and the best solution in $P(t)$
			\ENDIF
            \STATE Update $t= t+1$
        \ENDWHILE
        \RETURN Best solution in $P(t)$
    \end{algorithmic}
\end{algorithm}

The algorithm works as follows: A population $P(t)$ is initialized and evaluated with $N$ random solutions at iteration $t=0$. Then, the main loop of the algorithm is executed until a stopping condition is met. In this paper, the stopping criterion is to evaluate a maximum number of solutions so that the algorithms can be compared in performance. In order to test the algorithms, an additional stopping criterion has been included when a solution in the population reaches a previously known lower bound of the distance. 

The loop of the algorithms start by selecting $N$ parents according to the binary tournament selection operator \cite{Blickle97acomparison}. A crossover operator is applied to two parents to generate a pair of new solutions with probability $p_c$. If they are not combined, the mutation operator acts on the parents to generate two mutated solutions. All the $N$ new solutions generated by either crossover or mutation form the population at the next iteration $P(t+1)$. After that, in the discrete case, we check if all solutions in $P(t+1)$ are valid. Finally, the solutions in $P(t+1)$ are evaluated. An elitism component is included before the next iteration starts: If no solution in $P(t+1)$ has a fitness equal or better than the best in $P(t)$, then then worst in $P(t+1)$ is replaced with the best in $P(t)$. Also, we include a reinitialization of $P(t+1)$ with $N$ new random solutions after $MaxReinit$ solution evaluations with no improvement in the fitness of the best solution found.

The CHC algorithm is an evolutionary algorithm whose initial version was proposed for binary encoding \cite{eshelman1991}. This algorithm holds a balance between genotypic diversity in the solutions of the population, and convergence to local optima. It is based on four main components: elitist selection, the HUX solution recombination operator, an incest prevention check to avoid the recombination of similar solutions, and a population reinitialization method when a local optimum is found. Later versions of this algorithm are proposed for real and permutation encoding in \cite{eshelman93, cordon2006, chcdiscrete1}.  The adaptation of this algorithm, described in Algorithm \ref{alg:chc}, it is mainly inspired in the proposals of \cite{cordon2006, chcdiscrete1}. 
\floatname{algorithm}{Algorithm}
\begin{algorithm}[ht]
    \caption{CHC Algorithm}\label{alg:chc}
    \begin{algorithmic}[1]
        \REQUIRE $N = \text{Even natural number with the size of the population}$
        \REQUIRE $\tau = \text{Crossover threshold update rate}$
        \STATE Set $t$ = 0
        \STATE Initialize the population $P(t)$ with $N$ random valid solutions
        \STATE Evaluate solutions in $P(t)$
        \STATE Set $d = \textit{Average distance of solutions in $P(t)$}$
        \STATE Set $dec = \tau \cdot \textit{Maximum distance of solutions in $P(t)$}$
        \WHILE{stopping criterion is not satisfied}
        \STATE Set $C(t)= \emptyset$
        \STATE Set $parents(1, \dots, N) = \text{random shuffle of solutions in $P(t)$}$
        \FOR{$i = 0, \dots, N/2-1$}
        \IF{$distance(parents(2i+1), parents(2i+2)) < d$}
        \STATE Set $c_1, c_2 = \text{solutions generated from crossover over $parents(2i+1)$ and $parents(2i+2)$}$
        \STATE Update $C(t)= C(t) \cup \{c_1, c_2\}$
        \ENDIF
        \ENDFOR
        \STATE Evaluate Solutions in $C(t)$
        \STATE Set $P(t+1)= \text{Best $N$ solutions in $C(t)\cup P(t)$}$
        \IF{$P(t) = P(t+1)$}
        \STATE Update $d= d-dec$
        \IF{$d\leq 0$}
        \STATE Initialize $P(t+1)$ with the best solution of $P(t)$ and $N-1$ random solutions.
        \STATE Evaluate new solutions in $P(t+1)$
        \STATE Set $d = \textit{Average distance of solutions in $P(t+1)$}$
        \STATE Set $dec =\tau \cdot \textit{Maximum distance of solutions in $P(t+1)$}$
        \ENDIF
        \ENDIF
        \STATE $t=t+1$
        \ENDWHILE
        \RETURN Best solution in $P(t)$
    \end{algorithmic}
\end{algorithm}

The distance between two solutions $x$, $y$ in the population is computed using the Hamming distance (lines 4-5 and 20-21 in Algorithm \ref{alg:chc}), and the decrement $dec$ of the crossover (lines 5 and 21) is updated as a percentage $\tau$ of the maximum Hamming distance between individuals in the population, where $\tau \in [0,1]$ is an update rate, an input parameter to the algorithm. In the \textit{CHC-Discrete} algorithm, if a non-valid solution $x$ is generated by crossover, then it is assigned with a fitness equals to infinity so that it is ensured they are not in the population for the next generation. The reconstruction procedure is not included in the \textit{CHC-Order} algorithm to not interfere with the components of the algorithm to balance exploration and search intensification.

The algorithm starts by initializing a population $P(t)$ with $N$ random solutions. Then, the average and maximum distances between all solutions are computed. The crossover threshold $d$ is assigned to the average distance, and a threshold update rate $dec$ is initialized to $\tau$ multiplied by the maximum distance. The main loop of the algorithm finishes when the aforementioned stopping condition is met. It works as follows: Firstly, the solutions in $P(t)$ are randomly shuffled and matched by pairs. These pairs of solutions are the parents to be combined. Then, the crossover operator is applied to each pair of parents to generate two offsprings, only if the distance between the two parents is not under the distance threshold $d$. If so, the offsprings are evaluated, and the population at the next iteration $P(t+1)$ is created containing the best $N$ solutions coming from $P(t)$ and the new generated solutions by crossover. If $P(t+1)$ is the same as $P(t)$, the crossover threshold $d$ is decreased by $dec$. Only when $d$ is zero or under zero, the population is reinitialized. In our work, the reinitialization procedure replaces $P(t+1)$ with $N-1$ randomly generated solutions, and the best solution in $P(t)$. The values $d$ and $dec$ are recalculated for this new population.

More information about the structure of a generational genetic algorithm and the CHC can be consulted in \cite{BackFogelMichalewicz1997, eshelman1991} as classic introductory texts.

\subsection{Traditional problem statement and solution: discrete representation}

As we have mentioned previously in the Introduction, the representation scheme used to the date in the literature is the discrete representation. 
Suppose a linear code $\mathcal{C}$ of length $n$ and dimension $k$ over the finite field with $q$ elements $\mathbb{F}_q$ is given by a generating matrix $G \in \mathbb{F}_q^{k \times n}$. Then, each codeword $c\in\mathcal{C}$ is obtained by the encoding process of certain message $m\in \mathbb{F}_q^k$ as $c= mG$. Thus, the traditional problem statement attempts to find a message $m=(m_1, m_2, ..., m_k) \in \mathbb{F}^k_q$, with $m\not = 0$, such that the Hamming weight of its corresponding codeword $mG$ is minimum. With this problem statement, the solution representation assigns a solution/chromosome $x$ to a corresponding message $m$ as phenotype, whilst the genotype is composed by $k$ decision variables or genes whose values are in the set $\mathbb{F}_q$. 
Assuming that the arithmetic over the finite field is in $\mathcal{O}(1)$, then the fitness calculation is in $\mathcal{O}(kn)$. 

As the solution representation is discrete, several known crossover and mutation operators for discrete representation \cite{Gwiazda2006} can be selected, i.e. one point, two points, uniform crossover, random mutation by gene, etc. We performed a pre-experimentation to find out the best crossover and mutation settings for our experiments, and we have concluded that the uniform crossover and random mutation with mutation probability by gene are the best choices for \emph{GGA-Discrete}, as it is discussed in the experimental section.  The binary tournament selection \cite{Blickle97acomparison} was selected as the parent selection operator in \textit{GGA-Discrete} for the same reason. The implemented crossover for \emph{CHC-Discrete}  is also the uniform crossover.

As commented before, we have not found proposals in the literature that solve the problem using finite fields greater than $\mathbb{F}_2$, so that they could use binary representation instead of discrete representation. The binary representation for specific problems over $\mathbb{F}_2$ have some advantages from a practical perspective besides the higher range of possibilities to choose crossover and mutation operators, since the addition and multiplication in  $\mathbb{F}_2$ can be easily implemented as logic XOR and AND operators, and this can accelerate the execution time significantly using bitwise arithmetics. In our case, since, to the best of our knowledge, this paper is the first one to address the problem using metaheuristics with non binary finite fields, we use discrete representation.

\subsection{A new problem statement and solution: order representation}

The result stated in Section \ref{sec:permRep} allows us to change the solution space of the traditional problem statement to the space of permutations of $n$ elements, $\mathcal{S}_n$, where $n$ is the length of the linear code under study. Our proposal does not rely on the encoding process of a message $s\in\mathbb{F}_q^n$. Actually, it is independent of the encoding process and, therefore, of the finite field size, and this property is very convenient in case of a large finite field. 

According to Theorem, the problem of finding the minimum distance $\operatorname{d}(\mathcal{C})$ of a linear code $\mathcal{C}$ can be stated as finding a permutation $x\in \mathcal{S}_n$ for which a row $r$ of the RREF of a given generating matrix $G_{k,n}$ of $\mathcal{C}$ has Hamming weight $\operatorname{w}(r)=\operatorname{d}(\mathcal{C})$, where $k$ is the dimension of $\mathcal{C}$. Under this statement, a solution/chromosome of the population is assigned to an unique phenotype $x$, a permutation of the $n$ columns of the generating matrix $G$, and contains $n$ decision variables or genotype whose values can be in $\{1, 2, ..., n\}$, where each number stands for a column position in the matrix. Given a permutation $x=(x_1, x_2, ..., x_n)$, a value $x_i$ means that a matrix column located initially at position $i$ changes its location to the column position $x_i$ in the permutation (see Example \ref{ex2}).

The fitness of a permutation $x\in \mathcal{S}_n$ is then computed as $$f(x)= min\{ \operatorname{w}(r_i): r_i \text{ a row of } RREF(G_{n,k}P_x) \},$$
where $P_x$ denotes the permution matrix of $x$, i.e. the minimum weight of all rows of the RREF of the matrix $G_{n,k}$ whose columns have been permuted according to $x$. An optimal solution $x^*$ is a permutation such that $f(x^*)= \operatorname{d}(\mathcal{C})$, so that $x^*= min_x \{f(x)\}$, and the problem is stated as a minimization problem. The classic scholar algorithm to calculate RREF \cite[Chapter 5]{Liesen/Mehrmann:2015} can be implemented with complexity $\mathcal{O}(kn^2)$, and the product $G_{k,n}P_x$ can be implemented in $\mathcal{O}(n)$ in the most efficient implementation (for example, reassignments of $n$ memory pointers to matrix columns in C/C++ language). Thus, the complexity of the fitness function is in $\mathcal{O}(kn^2)$. It is worth mentioning that this theoretical complexity is far worse than the fitness calculation of the classic problem statement explained in the previous section. However, as we show in the experiments, this increase in complexity is well paid out with a substantial reduction in solution evaluations/algorithm iterations to find the optimal solution, which can be even reached with a random search in small problem instances, and to overcome local optima.

Regarding the genetic algorithm operators, classic crossover procedures do not consider the finite group structure of the set of all permutations $\mathcal{S}_n$. Intuitively, for a column permutation of a code matrix, the more non pivot columns are permuted to the first column positions of the code matrix, the better fitness it could have. Therefore, one could expect that the composition of permutations with good fitness, may produce a solution with better fitness. Additionally, since two (or more) random permutations probably form a generator system \cite[Theorem 1]{dixon}, the whole space of solutions is reached by their composition. In this paper we propose the algebraic crossover $AX_m$ with $m>1$: Given $m$ solutions from the population $\{x_1, x_2,\ldots ,x_m\}$, we construct the set of permutations given by their composition in any order, that is,
\[
\mathcal{T}=\{x_{\tau{(1)}}\circ x_{\tau{(2)}}\circ \cdots \circ x_{\tau{(r)}} \text{ such that }\tau\in \mathcal{S}_m\}.
\]
From this set, we select the $m$ solutions with lower image under $f(x)$ (that is, with better fitness) which replace the original $m$ solutions. Therefore, taking \(m\) a divisor of \(n\), the algebraic crossover operator $AX_m$ partitions the population into subsets of $m$ elements and, for each subset, with a given probability $p_c$, it recombines the elements as described above. Obviously, the crossover complexity increases exponentially with the parameter $m$, although the exploration of the space also grows. In the experiments performed in this paper, we use a fixed value $m=2$ to generate two new solutions from two parent solutions.  The same crossover is implemented in the \emph{CHC-Order} algorithm.

In the mutation step, we follow a standard mutation operator: the 2-swap mutation, i.e. the composition with a transposition. Nevertheless, permuting two columns which are beyond the last pivot does not modify the fitness. 
For reasons of efficiency, we simply choose a column from the first $k$ columns and other from the remaining $n-k$ columns in the generating matrix. The mutation operator is then applied to those solutions that were not recombined with the crossover operator.

\begin{example}\label{ex2}
Let us show the evolution of a population of chromosomes during the \emph{GGA-Order} algorithm execution in order to illustrate the modules described in this subsection. Let $\mathcal{C}$ be the $[4,10]_4$-linear code with generating matrix
$$G_{4, 10}=\left(\begin{array}{cccccccccc}
a & 1 & a + 1 & 0 & 0 & a & 1 & 1 & a + 1 & a \\
a & 0 & a & 1 & 1 & a + 1 & a & a + 1 & a + 1 & a \\
0 & 1 & a + 1 & 1 & a + 1 & a & 0 & a & a & a + 1 \\
a & 1 & 1 & 0 & 0 & 1 & a + 1 & a & a + 1 & 1
\end{array}\right).$$
Suppose that we have a population  formed by the following four chromosomes.
$$\begin{array}{ccc}
c_1 =& [6, 4, 3, 9, 7, 10, 2, 1, 8, 5] & \to \text{fitness} = 4\\
c_2 =& [9, 3, 7, 10, 1, 4, 5, 2, 8, 6] & \to \text{fitness} = 6\\
c_3 =& [1, 8, 9, 7, 6, 2, 3, 10, 4, 5] & \to \text{fitness} = 5\\
c_4 =& [2, 9, 8, 3, 4, 10, 6, 5, 7, 1] & \to \text{fitness} = 5.
\end{array}$$
The list associated to each chromosome describes the permutation as a map. For instance, $c_1$ maps 4 to 9, since $c_1[4]=9$, or $c_3$ maps 10 to 5, since $c_1[10]=5$. The fitness of each permutation is computed as described un Example \ref{ex:rep}.

Let us suppose that the probability of a crossover is $p_c=0.5$. Then we generated $d_1$ and $d_2$ from crossover over parents $c_1$ and $c_2$, whilst this operator is not applied to $c_3$ and $c_4$. Hence 
$$\begin{array}{ccc}
d_1 = c_1 \circ c_2 = & [4, 10, 7, 8, 5, 6, 3, 9, 2, 1] & \to \text{fitness} = 5\\
d_2 = c_2 \circ c_1 = & [8, 3, 2, 5, 6, 9, 7, 4, 1, 10] & \to \text{fitness} = 5
\end{array}$$
On the other hand, we applied a mutation (composition with a transposition) to $c_3$ and $c_4$ in order to obtain offsprings $d_3$ and $d_4$. For instance, let us suppose
$$\begin{array}{lll}
d_3 = (1,10) \circ c_3 = & [5, 8, 9, 7, 6, 2, 3, 10, 4, 1] & \to \text{fitness} = 5\\
d_4 = (3,6) \circ c_4  = & [2, 9, 10, 3, 4, 8, 6, 5, 7, 1] & \to \text{fitness} = 4
\end{array}$$
The evolved population is then $\{d_1,d_2,d_3,d_4\}$. Observe that elitism is not applied, since $d_4$ is as good as $c_1$.
%
%
%
\end{example}

\section{Experimentation}\label{Section3}

We consider \(30\) different datasets to test our approach. The first \(20\) are generating matrices of linear codes over \(\mathbb{F}_8\), from the database \textsc{Magma} \cite{Magma}, a specific software for computer algebra, using the command \texttt{BKLC} (an acronym of \emph{Best Known Linear Code}). They are labelled as \((n,k,d_1-d_2)\), where \(n\) is the length, \(k\) the dimension, \(d_1\) a lower bound and \(d_2\) an upper bound for the distance. The other \(10\) datasets are generating matrices of codes over \(\mathbb{F}_2\). Five of them are the generating matrices of narrow sense BCH (Bose–Chaudhuri–Hocquenghem) taken from \textsc{Magma} with the command \texttt{BCHCode}. They are labelled \(\BCH(n,k,\delta)\), where \(n,k\) are the length and the dimension, and \(\delta\) is the designed distance, a lower bound established during the code design. Finally, remaining five datasets are the generating matrices of EQR (Extended Quadratic Residues) codes. These matrices were obtained adding a parity-check bit to the generating matrices of suitable Quadratic Residues codes obtained from \textsc{Magma} via the command \texttt{QRCode}. They are labelled \(\EQR(n,k,d)\) where \(n,k\) denote the length and dimension, and \(d\) the code distance obtained from \cite{Leon1988}.

\subsection{Experimental design}

\begin{table}[!ht]
	\centering
\begin{small}
	\begin{tabular}{|c||c||c||c||c||c|}
		\hline
		 \textbf{Field} & \((n,k,d_1-d_2)\) & \((n,k,d_1-d_2)\) & \((n,k,d_1-d_2)\) & \((n,k,d_1-d_2)\) & \((n,k,d_1-d_2)\)\\
		\hline
$\mathbb{F}_8$ & \((30 , 16 , 10-13)\) & 
\((30 , 18 , 9-11)\) & 
\((45 , 22 , 15-21)\) & 
\((45 , 24 , 14-19)\) & 
\((45 , 26 , 12-17)\) \\
$\mathbb{F}_8$ & \((45 , 28 , 11-15)\) & 
\((60 , 28 , 21-28)\) & 
\((60 , 30 , 20-27)\) & 
\((60 , 32 , 19-25)\) & 
\((60 , 34 , 17-23)\) \\
$\mathbb{F}_8$ & \((75 , 30 , 28-40)\) & 
\((75 , 35 , 24-35)\) & 
\((75 , 40 , 20-31)\) & 
\((75 , 45 , 17-27)\) & 
\((90 , 19 , 49-63)\) \\
$\mathbb{F}_8$ & \((90 , 50 , 21-35)\) & 
\((90 , 60 , 16-26)\) & 
\((90 , 70 , 11-17)\) & 
\((130 , 85 , 23-40)\) & 
\((130 , 95 , 18-30)\)\\
		\hline
		  & \(\EQR(n,k,d)\) & \(\EQR(n,k,d)\) & \(\EQR(n,k,d)\) & \(\EQR(n,k,d)\) & \(\EQR(n,k,d)\)\\
\hline
$\mathbb{F}_2$ & \(\EQR(272,176,40)\) & 
\(\EQR(338,169,40)\) & 
\(\EQR(368,184,48)\) & 
\(\EQR(432,216,48)\) & 
\(\EQR(440,220,48)\) \\
		\hline
		  & \(\BCH(n,k,\delta)\) & \(\BCH(n,k,\delta)\) & \(\BCH(n,k,\delta)\) & \(\BCH(n,k,\delta)\) & \(\BCH(n,k,\delta)\)\\
\hline
$\mathbb{F}_2$ & \(\BCH(511,76,171)\) & 
\(\BCH(511,103,123)\) & 
\(\BCH(511,121,117)\) & 
\(\BCH(511,166,95)\) & 
\(\BCH(511,184,91)\) \\
         \hline
	\end{tabular}
\end{small}
	\caption{Description of the codes selected for the experimentation}
	\label{table:codesettings}
\end{table}

The experiments were conducted in two stages: In the first one (Section \ref{exp:classicRepr}), the order representation proposed in this work is compared with the classic discrete representation. The objective of such study is to test if the genetic algorithms using the order representation outperform the ones with discrete representation, regarding the quality of the results for codes over $\mathbb{F}_8$. The second phase (Section \ref{exp:previousAppr}) compares the results of our proposal with some state-of-the-art methods existing in the literature. More specifically, the baseline method employed for comparison is published in \cite{askali13}. As no previous study was found over a field greater than $\mathbb{F}_2$, only the EQR and BCH codes are contemplated in this stage. Section \ref{exp:discussion} discusses the results and describes the conclusions obtained.

Finally, \(100\) experiments with different random seed were performed to obtain a set of results that can be analyzed statistically. We used a desktop equipped with CPU Intel Core I7, 8GB RAM, and Ubuntu 18.04 O.S. as experiment host.

\subsection{Analysis of results over $\mathbb{F}_8$}\label{exp:classicRepr}

The experimentation was performed with two different algorithms: GGA and CHC. Both were implemented using discrete and order representation, as described in Section \ref{sec:algs}, in separate experiments, so that there are four different configurations to be tested: GGA-Discrete, GGA-Order, CHC-Discrete, and CHC-Order. The reason of including the CHC algorithms in the experimentation is their ability to overcome local optima due to the their component design to establish a balance between diversity and convergence. This strategy is beneficial to obtain better solutions by avoiding premature convergence of traditional GGAs.

The algorithms' parameters were set after a trial-and-error procedure. As the number of datasets is high, these parameters were not fitted specifically for each problem, but as those that can provide good results in average. Additional experiments were conducted individually for the algorithms with the worst performance to find better parameter settings, with no success. 

The parameters used in the experimentation are shown in Table \ref{table:parametersF8}. The row \textit{Mutation probability by gene} is describes the probability of mutation of each coordinate in a solution. The row \textit{Evaluations} stands for the main stopping criterion of each algorithm (i.e. to reach a maximum number of evaluations). The row \textit{Reinitialization} shows the number of solutions evaluated with no improvements required to reinitialize the population in GGA proposals. A secondary stopping criterion was set for all algorithms: if the true known distance of the code being evaluated (or a lower bound) is found, the algorithm stops even if the main stopping criterion is not satisfied yet. 

\begin{table}[!ht]
	\centering
	\begin{tabular}{|c||c||c||c||c|}
		\hline
		 \textbf{Parameter} & \textbf{GGA-Discrete} & \textbf{CHC-Discrete} & \textbf{GGA-Order} & \textbf{CHC-Order}\\
		\hline
Population size & 400 & 400 & 400 & 400\\
Crossover & Uniform & Uniform & \(AX_2\) & \(AX_2\) \\
Mutation & Random & -- & 2-swap & -- \\
Crossover probability & 0.7 & -- & 0.8 & --\\
Mutation probability by gene & 0.01 & -- & -- & -- \\
Evaluations & 500000 & 500000 & 500000 & 500000 \\
Reinitialization & 100000 & -- & 100000 & -- \\
         \hline
	\end{tabular}
	\caption{Algorithms' parameters for experimentation with codes over $\mathbb{F}_8$}
	\label{table:parametersF8}
\end{table}

Tables \ref{table:resultsDiscreteF8} and \ref{table:resultsOrderF8} show the results obtained for each dataset containing a code over $\mathbb{F}_8$. The columns \textbf{Dataset} display the parameters of the corresponding best known linear codes according to \cite{codetables}. Columns \textbf{Best} describe the minimum weight found by each algorithm, and the subindices indicate the number of experiments that provided such result. Columns \textbf{Worst} introduce the worst minimum weight obtained. Columns \textbf{Mean} show the average of the minimum weights obtained in the \(100\) experiments. Finally, Columns \textbf{Time/Eval.} print the average time in seconds of each experiment, together with the average number of evaluations required to finish it.

\begin{table}[!ht]
	\centering
\begin{small}
	\begin{tabular}{|l|c|c|c|c|c|c|c|c|}
		\hline
		& \multicolumn{4}{c|}{\textbf{CHC (discrete)}} & \multicolumn{4}{c|}{\textbf{GGA (discrete)}}  \\ \hline
		\textbf{Dataset} & \textbf{Best} & \textbf{Worst} & \textbf{Mean} & \textbf{Time / Eval.} & \textbf{Best} & \textbf{Worst} & \textbf{Mean} & \textbf{Time / Eval.}\\ 
		\hline
		\((30,16,10-13)\) & $10_{(100)}$ & $10$ & $10_{(1)}$ & $0.136_{(4)}  / 7117.3$ & $10_{(100)}$ & $10$ & $10_{(1)}$ & $0.015_{(3)}  / 6074.6$\\
		\((30,18,9-11)\) & $9_{(100)}$ & $9$ & $9_{(1)}$ & $0.152_{(4)}  / 7745.6$ & $9_{(100)}$ & $9$ & $9_{(1)}$ & $0.02_{(3)}  / 6742.8$\\
		\((45,22,15-21)\) & $15_{(90)}$ & $16$ & $15.1_{(2)}$ & $1.865_{(4)}  / 148352.3$ & $15_{(100)}$ & $15$ & $15_{(1)}$ & $0.304_{(3)}  / 85878.5$\\
		\((45,24,14-19)\) & $14_{(100)}$ & $14$ & $14_{(1)}$ & $0.245_{(4)}  / 15373.6$ & $14_{(100)}$ & $14$ & $14_{(1)}$ & $0.032_{(3)}  / 10307$\\
		\((45,26,12-17)\) & $12_{(100)}$ & $12$ & $12_{(1)}$ & $0.717_{(4)}  / 43823.4$ & $12_{(100)}$ & $12$ & $12_{(1)}$ & $0.193_{(3)}  / 49539.9$\\
		\((45,28,11-15)\) & $11_{(100)}$ & $11$ & $11_{(1)}$ & $0.309_{(4)}  / 13732$ & $11_{(100)}$ & $11$ & $11_{(1)}$ & $0.039_{(3)}  / 10605.2$\\
		\((60,28,21-28)\) & $21_{(80)}$ & $22$ & $21.2_{(2)}$ & $2_{(4)}  / 164512.4$ & $21_{(81)}$ & $22$ & $21.2_{(2)}$ & $0.867_{(3)}  / 196025.9$\\
		\((60,30,20-27)\) & $20_{(100)}$ & $20$ & $20_{(1)}$ & $0.4_{(4)}  / 21209.3$ & $20_{(100)}$ & $20$ & $20_{(1)}$ & $0.062_{(3)}  / 13229.9$\\
		\((60,32,19-25)\) & $19_{(100)}$ & $19$ & $19_{(1)}$ & $0.26_{(4)}  / 19573.3$ & $19_{(100)}$ & $19$ & $19_{(1)}$ & $0.053_{(3)}  / 12651.8$\\
		\((60,34,17-23)\) & $17_{(98)}$ & $18$ & $17_{(1)}$ & $1.929_{(4)}  / 122798.7$ & $17_{(100)}$ & $17$ & $17_{(1)}$ & $0.421_{(3)}  / 78914.9$\\
		\((75,30,28-40)\) & $28_{(2)}$ & $32$ & $30.1_{(2)}$ & $3.995_{(4)}  / 309250.5$ & $29_{(3)}$ & $32$ & $30.4_{(3)}$ & $2.754_{(3)}  / 500000$\\
		\((75,35,24-35)\) & $24_{(1)}$ & $28$ & $27.6_{(2)}$ & $3.685_{(3)}  / 228937$ & $26_{(2)}$ & $28$ & $27.6_{(2)}$ & $3.369_{(2)}  / 500000$\\
		\((75,40,20-31)\) & $20_{(7)}$ & $24$ & $22.4_{(2)}$ & $3.968_{(4)}  / 262879.6$ & $20_{(2)}$ & $25$ & $22.8_{(3)}$ & $1.005_{(3)}  / 144057$\\
		\((75,45,17-27)\) & $18_{(97)}$ & $20$ & $18.1_{(2)}$ & $10.35_{(2)}  / 500000$ & $18_{(88)}$ & $20$ & $18.2_{(3)}$ & $4.208_{(1)}  / 500000$\\
		\((90,19,49-63)\) & $49_{(6)}$ & $54$ & $51.9_{(2)}$ & $2.151_{(4)}  / 260640.5$ & $49_{(5)}$ & $52$ & $51.8_{(2)}$ & $1.56_{(3)}  / 285425.6$\\
		\((90,50,21-35)\) & $27_{(17)}$ & $29$ & $27.8_{(2)}$ & $9.174_{(2)}  / 500000$ & $25_{(1)}$ & $29$ & $27.8_{(2)}$ & $4.77_{(1)}  / 500000$\\
		\((90,60,16-26)\) & $18_{(1)}$ & $21$ & $20.2_{(2)}$ & $12.475_{(2)}  / 500000$ & $19_{(29)}$ & $21$ & $20.1_{(2)}$ & $5.71_{(1)}  / 500000$\\
		\((90,70,11-17)\) & $12_{(2)}$ & $14$ & $13_{(2)}$ & $17.405_{(4)}  / 500000$ & $11_{(1)}$ & $14$ & $13_{(2)}$ & $1.554_{(3)}  / 120408$\\
		\((130,85,23-40)\) & $29_{(1)}$ & $32$ & $31.3_{(3)}$ & $20.469_{(2)}  / 500000$ & $29_{(2)}$ & $32$ & $31_{(2)}$ & $10.538_{(1)}  / 500000$\\
		\((130,95,18-30)\) & $22_{(8)}$ & $24$ & $23.4_{(3)}$ & $23.758_{(2)}  / 500000$ & $22_{(21)}$ & $24$ & $22.9_{(2)}$ & $11.745_{(1)}  / 500000$\\
		\hline
	\end{tabular}
\end{small}
\caption{Results of CHC-Discrete and GGA-Discrete over $\mathbb{F}_8$ codes}
\label{table:resultsDiscreteF8}
\end{table}

\begin{table}[!ht]
	\centering
\begin{small}
	\begin{tabular}{|l|c|c|c|c|c|c|c|c|}
		\hline
		& \multicolumn{4}{c|}{\textbf{CHC (order)}} & \multicolumn{4}{c|}{\textbf{GGA (order)}}  \\ \hline
		\textbf{Dataset} & \textbf{Best} & \textbf{Worst} & \textbf{Mean} & \textbf{Time / Eval.} & \textbf{Best} & \textbf{Worst} & \textbf{Mean} & \textbf{Time / Eval.}\\ \hline
		\((30,16,10-13)\) & $10_{(100)}$ & $10$ & $10_{(1)}$ & $0.004_{(2)}  / 1.6$ & $10_{(100)}$ & $10$ & $10_{(1)}$ & $0.003_{(1)}  / 1.6$\\
		\((30,18,9-11)\) & $9_{(100)}$ & $9$ & $9_{(1)}$ & $0.005_{(2)}  / 2$ & $9_{(100)}$ & $9$ & $9_{(1)}$ & $0.004_{(1)}  / 2$\\
		\((45,22,15-21)\) & $15_{(100)}$ & $15$ & $15_{(1)}$ & $0.01_{(2)}  / 57.8$ & $15_{(100)}$ & $15$ & $15_{(1)}$ & $0.009_{(1)}  / 57.8$\\
		\((45,24,14-19)\) & $14_{(100)}$ & $14$ & $14_{(1)}$ & $0.011_{(2)}  / 13.4$ & $14_{(100)}$ & $14$ & $14_{(1)}$ & $0.01_{(1)}  / 13.4$\\
		\((45,26,12-17)\) & $12_{(100)}$ & $12$ & $12_{(1)}$ & $0.011_{(2)}  / 9.5$ & $12_{(100)}$ & $12$ & $12_{(1)}$ & $0.01_{(1)}  / 9.5$\\
		\((45,28,11-15)\) & $11_{(100)}$ & $11$ & $11_{(1)}$ & $0.011_{(2)}  / 1.9$ & $11_{(100)}$ & $11$ & $11_{(1)}$ & $0.01_{(1)}  / 1.9$\\
		\((60,28,21-28)\) & $21_{(100)}$ & $21$ & $21_{(1)}$ & $0.019_{(2)}  / 42.8$ & $21_{(100)}$ & $21$ & $21_{(1)}$ & $0.017_{(1)}  / 42.8$\\
		\((60,30,20-27)\) & $20_{(100)}$ & $20$ & $20_{(1)}$ & $0.019_{(2)}  / 5$ & $20_{(100)}$ & $20$ & $20_{(1)}$ & $0.018_{(1)}  / 5$\\
		\((60,32,19-25)\) & $19_{(100)}$ & $19$ & $19_{(1)}$ & $0.021_{(2)}  / 8.3$ & $19_{(100)}$ & $19$ & $19_{(1)}$ & $0.02_{(1)}  / 8.3$\\
		\((60,34,17-23)\) & $17_{(100)}$ & $17$ & $17_{(1)}$ & $0.021_{(2)}  / 18.4$ & $17_{(100)}$ & $17$ & $17_{(1)}$ & $0.021_{(1)}  / 18.4$\\
		\((75,30,28-40)\) & $28_{(100)}$ & $28$ & $28_{(1)}$ & $0.029_{(2)}  / 19.1$ & $28_{(100)}$ & $28$ & $28_{(1)}$ & $0.028_{(1)}  / 19.1$\\
		\((75,35,24-35)\) & $24_{(100)}$ & $24$ & $24_{(1)}$ & $0.12_{(1)}  / 1269.2$ & $24_{(100)}$ & $24$ & $24_{(1)}$ & $0.134_{(1)}  / 1716.2$\\
		\((75,40,20-31)\) & $20_{(100)}$ & $20$ & $20_{(1)}$ & $0.037_{(2)}  / 5.9$ & $20_{(100)}$ & $20$ & $20_{(1)}$ & $0.036_{(1)}  / 5.9$\\
		\(\mathbf{(75,45,17-27)}\) & $\mathbf{15}_{(100)}$ & $\mathbf{15}$ & $\mathbf{15}_{(1)}$ & $11.695_{(2)}  / 72438.6$ & $\mathbf{15}_{(100)}$ & $\mathbf{15}$ & $\mathbf{15}_{(1)}$ & $12.713_{(2)}  / 77303.3$\\
		\((90,19,49-63)\) & $49_{(100)}$ & $49$ & $49_{(1)}$ & $0.179_{(1)}  / 3279.7$ & $49_{(100)}$ & $49$ & $49_{(1)}$ & $0.252_{(2)}  / 5249.6$\\
		\((90,50,21-35)\) & $22_{(34)}$ & $24$ & $22.7_{(1)}$ & $41.898_{(4)}  / 267888.3$ & $22_{(28)}$ & $24$ & $22.8_{(1)}$ & $38.437_{(3)}  / 239954.9$\\
		\((90,60,16-26)\) & $16_{(100)}$ & $16$ & $16_{(1)}$ & $13.265_{(3)}  / 78140$ & $16_{(99)}$ & $17$ & $16_{(1)}$ & $15.42_{(3)}  / 86775.2$\\
		\((90,70,11-17)\) & $11_{(100)}$ & $11$ & $11_{(1)}$ & $0.066_{(1)}  / 48.9$ & $11_{(100)}$ & $11$ & $11_{(1)}$ & $0.066_{(2)}  / 48.9$\\
		\((130,85,23-40)\) & $23_{(1)}$ & $26$ & $24.9_{(1)}$ & $63.325_{(3)}  / 135700$ & $23_{(2)}$ & $26$ & $25_{(1)}$ & $145.695_{(4)}  / 293522$\\
		\((130,95,18-30)\) & $18_{(88)}$ & $19$ & $18.1_{(1)}$ & $84.811_{(3)}  / 185147.7$ & $18_{(80)}$ & $19$ & $18.2_{(1)}$ & $87.796_{(3)}  / 177852.8$\\
		\hline
	\end{tabular}
\end{small}
\caption{Results of CHC-Order and GGA-Order over $\mathbb{F}_8$ codes}
\label{table:resultsOrderF8}
\end{table}

Statistical tests were applied over the mean of the upper bounds obtained in the experiments for each data set (columns \textbf{Mean}), to know if there are significant differences between the average behavior of the algorithms. For each dataset, all algorithms were sorted by its mean in ascending order, and the non-parametric Wilcoxon test with 95\% of confidence level was applied to compare them in the aforementioned sequence. According to the results, we ranked the algorithms from \(1\) to \(4\), where \(1\) is the best algorithm, and \(4\) the worst one. The rank order of each algorithm is depicted in columns \textbf{Mean} as subindices. We remark that two algorithms with the same rank says that there is no statistical difference between their performance.

According to the statistical tests, Table \ref{table:resultsOrderF8} shows that CHC-Order and GGA-Order return the best results, and there are no statistical difference between both algorithms in terms of performance. In fact, both are able to output the true distance in all datasets except for the code $(90, 50, 21-35)$, where a new upper bound, \(22\), is stablished. In particular, the optimum was reached in all experiments for \(17\) datasets (CHC-order) and \(16\) datasets (GGA-order). The worst behavior was obtained by both algorithms for the code $(130, 85, 23-40)$, where the optimum was reached in \(1\) and \(2\) experiments for CHC-order and GGA-order, respectively. Something similar happens with the code $(130, 95, 18-30)$. In this case, the true minimum distance was found in \(88\) experiments by CHC-order, and in \(80\) experiments by GGA. Both algorithms return a distance of \(19\) in the worst case.

We highlight the results for the code $(70, 45, 17-27)$ for which we found out an upper bound of the distance that equals \(15\). According to \cite{codetables}, the code is constructed using the $(77, 46, 18)$ linear code proposed in \cite{XingLing2000}. As our result falls below the claimed lower bound, we made additional analyses and compute a codeword whose weight equals \(15\) both in $(70, 45, 17-27)$ and $(77, 46, 18)$ codes. Both generating matrices are provided by the webpage \cite{codetables} and \textsc{Magma} \cite{Magma}. Regarding the discovered inaccuracy, Section \ref{sec:magma} shows our findings and describes explicitly some codewords of weight \(15\).

In contrast, table \ref{table:resultsDiscreteF8} shows that CHC-Discrete and GGA-Discrete provided the worst performance, and there are significant differences being compared with the algorithms that use the order representation. We may see that discrete representation algorithms are able to find the optimum code distance for small codes, up to length $n=75$. For larger codes, they get trapped into local optima. In fact, discrete representation-based approaches suffer from high selective pressure and get trapped into local optima near to the $(0, \dots, 0) \in \mathbb{F}_q^k$ trivial solution, which is not a valid solution, and this makes difficult to explore the solution space suitably. To give support to this statement, Figure \ref{fig:diversityF8} plots the evolution of the diversity during the execution of the algorithms GGA-Discrete and GGA-Order (right-hand figure), and GGA-Discrete and CHC-Discrete (left-hand figure), for the best solutions found for the code $(90, 50, 21)$. The diversity is shown in log scale. If we compare GGA-Discrete and GGA-Order, we may see how diversity in GGA-Discrete drops fast to near zero, and only returns to high values when the population is reinitialized. However, GGA-Order remains constant through all the iterations. 
On the other hand, if we compare CHC-Discrete \textit{vs} GGA-Discrete, they show a similar behavior. The main difference between the two graphics displayed in Figure \ref{fig:diversityF8}
is that the number of iterations in CHC-Discrete is higher, and the diversity drop is softer. This behavior is explained by the components of CHC to find a balance between diversity and convergence. However, according to Table \ref{table:resultsDiscreteF8}, it is not enough to overcome the high selective pressure of the trivial solution, and no statistical difference can be found regarding performance.

\begin{figure}[!ht]
	\centering
	\begin{subfigure}
		{\epsfig{file = 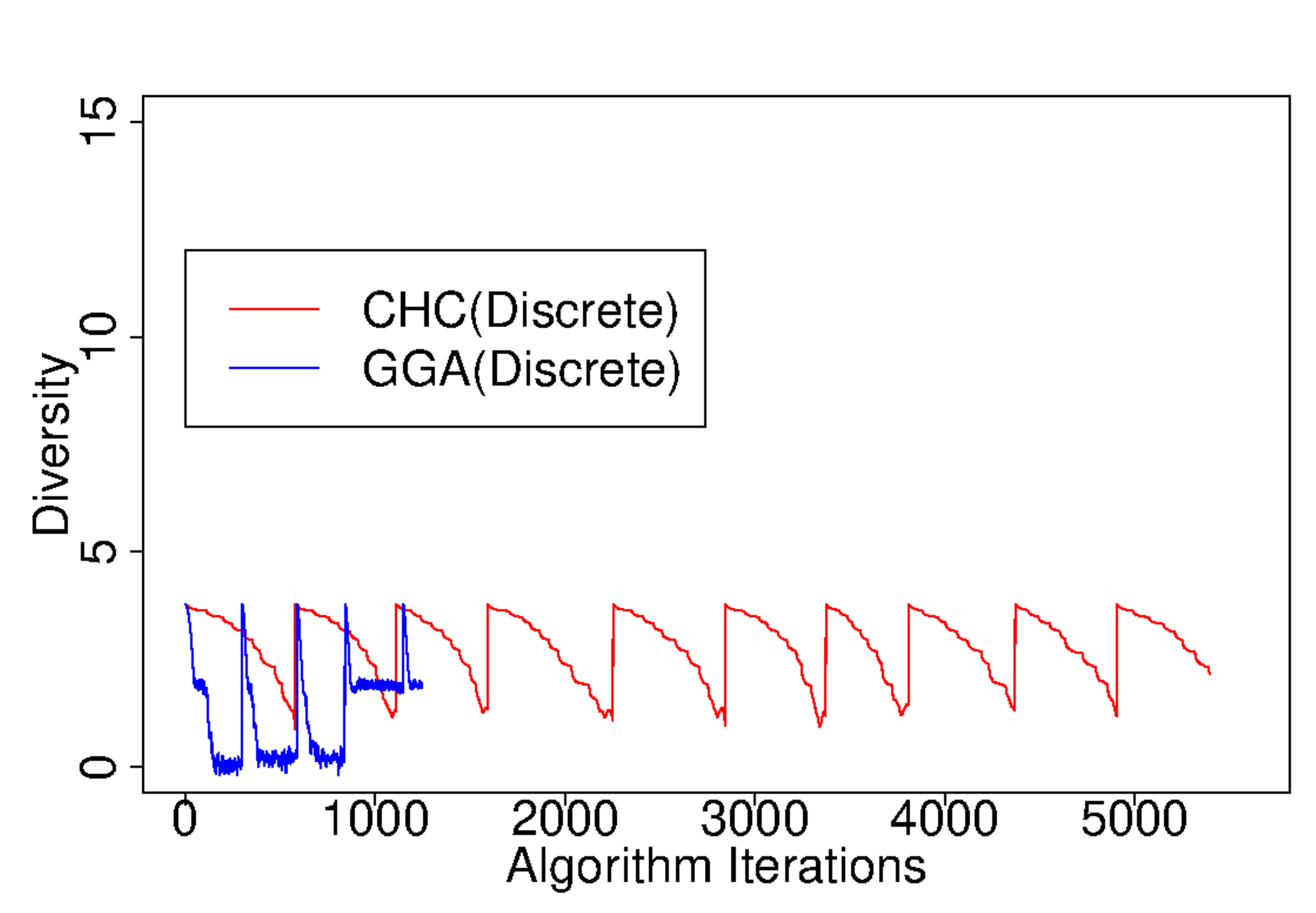, width = 7cm}
			\label{fig:GGACHCDiscreteGF8}}
	\end{subfigure}	
	\rulesep
	\begin{subfigure}
		{\epsfig{file = 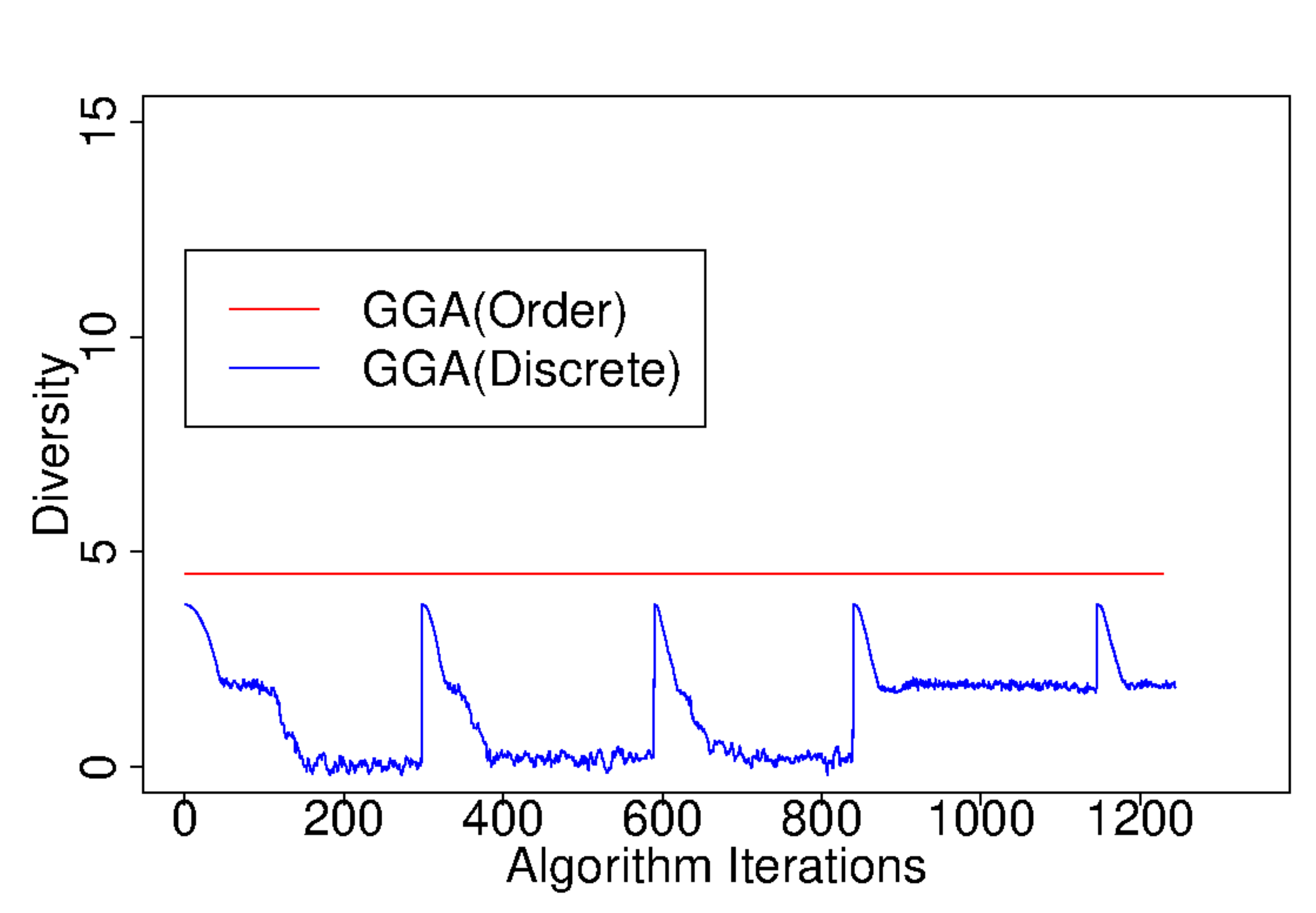, width = 7cm}
			\label{fig:GGAGGAGF8}}
	\end{subfigure}
	\caption{Evolution of the population diversity: GGA-Discrete vs CHC-Discrete (left), and GGA-Order vs GGA-Discrete (right)}
	\label{fig:diversityF8}
\end{figure}

To finish the analysis of the results, we are also interested in the comparison of the time complexities. The four algorithms were sorted by their average execution time in ascending order, and the Wilcoxon test with \(95\%\) of confidence level was applied to rank them all. The subindices in columns \textbf{Time/Eval.} of Tables \ref{table:resultsDiscreteF8} and \ref{table:resultsOrderF8} show the rank of each algorithm, where a lower value means that the algorithm is faster. The main bottleneck between discrete and order approaches is the fitness evaluation function, which are in $\mathcal{O}(kn)$ and $\mathcal{O}(k^2 n)$ respectively. Under this consideration, it could be expected that GGA-Discrete and CHC-Discrete were faster than GGA-Order and CHC-Order. This occurs only for larger codes, where the number of evaluations is high. In the case of small codes, discrete representation approaches are worse than the order representation ones in terms of execution time. In fact, if we focus on GGA-Order and CHC-Order for small codes, the average number of evaluations required to obtain the distance is even lower than the size of the population, meaning that those optimal solutions could be easily found by a random search. This fact provides an additional interest in the proposed method in this manuscript, since it shows empirically that using the RREF simplifies the search of the minimum distance at an experimental level.

To conclude this section, we remark the following outcomes: (a) The proposed order representation scheme, in conjunction with the strategy of computing the RREF of the generating matrix, help to overcome the limitations of the classic discrete representation in terms of falling into local optima; (b) despite the higher time complexity to compute the RREF, this strategy changes the search space which eases finding optimal solutions, even with a random search for small problems, as it is suggested by analyzing the number of evaluations needed to reach the optimum, printed in Tables \ref{table:resultsDiscreteF8} and \ref{table:resultsOrderF8}; and (c) the experimental results suggest that designing  metaheuristics with permutation representation can be an efficient and competitive tool to calculate minimum distance. In the codes studied, for which the minimum distance was unknown, our approach was able to find the true code distance in \(18\) of the \(20\) datasets, and a new upper bound was found for the remaining. It was also able to detect an inaccuracy in the code $(75,45,17-27)$, establishing that \(15\) is a new upper bound of its distance, see Section \ref{sec:magma}.

\subsection{Analysis of results over $\mathbb{F}_2$.}\label{exp:previousAppr}

In this section we analyze the results for BCH and EQR codes over $\mathbb{F}_2$. We also set the parameters from Table \ref{table:parametersF8}, except the maximum number of evaluations (stopping criterion) and the size of the population, fixed after a trial-and-error procedure. The number of evaluations to stop the algorithms was set to $1000000$ in the case of BCH codes, and $500000$ for EQR codes. The size of the population was set to $200$ for $\BCH(511, 103, 123)$, $\BCH(511, 166, 95)$ and $\BCH(511, 184, 91)$, since in a preliminary experimentation we observed a high space exploration but low exploitation. The size of the population was $400$ for the other codes. Finally, other changes regarding crossover probability, crossover method, reinitialization, etc., did not provide improvements with respect to accuracy of results, and they were set as in Table \ref{table:parametersF8}. 

We remark that the designed distance of a BCH codes is a lower bound of the minimum distance, although it is not ensured that they coincide. This is the case of $\BCH(511, 103, 123)$, with designed distance \(123\) and minimum distance \(127\), and $\BCH(511, 121, 117)$ (designed distance \(117\), minimum distance \(119\)). The minimum distances for these codes are computed in \cite[Table I]{AugotCharpinSendrier1992}.

\begin{table}[!ht]
	\centering
\begin{scriptsize}
	\begin{tabular}{|l|c|c|c|c|c|c|c|c|}
		\hline
		& \multicolumn{4}{c|}{\textbf{CHC (order)}} & \multicolumn{4}{c|}{\textbf{GGA (order)}}  \\ \hline 
		\textbf{Dataset} & \textbf{Best} & \textbf{Worst} & \textbf{Mean} & \textbf{Time / Eval.} & \textbf{Best} & \textbf{Worst} & \textbf{Mean} & \textbf{Time / Eval.}\\ \hline 
		\(\BCH(511,103,123)\) & $127_{(2)}$ & $144$ & $132.5_{(1)}$ & $1340.28_{(1)}  / 1000000$ & $127_{(4)}$ & $144$ & $133.6_{(1)}$ & $1379.63_{(2)}  / 1000000$\\
		\(\BCH(511,121,117)\) & $120_{(3)}$ & $139$ & $133.3_{(1)}$ & $1725.28_{(2)}  / 1000000$ & $120_{(1)}$ & $139$ & $134.1_{(1)}$ & $1712.15_{(1)}  / 1000000$\\
		\(\BCH(511,166,95)\) & $95_{(1)}$ & $119$ & $114.3_{(1)}$ & $2912.57_{(2)}  / 519668$ & $95_{(1)}$ & $119$ & $114.4_{(1)}$ & $2780.14_{(1)}  / 196196$\\
		\(\BCH(511,76,171)\) & $171_{(100)}$ & $171$ & $171_{(1)}$ & $0.627_{(1)}  / 4424.6$ & $171_{(100)}$ & $171$ & $171_{(1)}$ & $0.55_{(1)}  / 254$\\
		\(\BCH(511,58,183)\) & $183_{(100)}$ & $183$ & $183_{(1)}$ & $0.245_{(2)}  / 805$ & $183_{(100)}$ & $183$ & $183_{(1)}$ & $0.21_{(1)}  / 1175.5$\\
		\(\EQR(272,136,40)\) & $40_{(100)}$ & $40$ & $40_{(1)}$ & $0.35_{(2)}  / 3604.8$ & $40_{(100)}$ & $40$ & $40_{(1)}$ & $0.327_{(1)}  / 4784.8$\\
		\(\EQR(338,169,40)\) & $40_{(94)}$ & $50$ & $40.3_{(1)}$ & $5.058_{(1)}  / 156803.8$ & $40_{(88)}$ & $50$ & $40.6_{(1)}$ & $11.551_{(2)}  / 170266.8$\\
		\(\EQR(368,184,48)\) & $48_{(13)}$ & $56$ & $54.2_{(1)}$ & $317.353_{(1)}  / 345111.2$ & $48_{(10)}$ & $60$ & $54.7_{(1)}$ & $801.654_{(2)}  / 246819$\\
		\(\EQR(432,216,48)\) & $48_{(11)}$ & $68$ & $62_{(1)}$ & $561.347_{(1)}  / 197260.3$ & $48_{(5)}$ & $68$ & $63.6_{(2)}$ & $1453.67_{(2)}  / 180514.4$\\
		\(\EQR(440,220,48)\) & $48_{(17)}$ & $72$ & $65.8_{(2)}$ & $526.696_{(2)}  / 217904.2$ & $48_{(82)}$ & $68$ & $50.9_{(1)}$ & $176.835_{(1)}  / 167620.9$\\
		\hline 
	\end{tabular} 
\end{scriptsize}
\caption{Results of CHC-Order and GGA-Order over $\mathbb{F}_2$ codes}
\label{table:resultsOrderF2}
\end{table}

Table \ref{table:resultsOrderF2} shows the results obtained for BCH and EQR codes using the algorithms CHC-Order and GGA-Order. Both algorithms reach the true minimum distance in all cases, except for $\BCH(511, 121, 117)$ whose distance is \(119\) according to \cite[Table I]{AugotCharpinSendrier1992}. In spite of this, they were able to find an upper bound of distance with the same error correcting capability.

If we focus on the number of experiments that reached the optimal code distance in the datasets, we may observe the same behavior for both algorithms: As the dimension $k$ increases, the probability of success for an experiment to reach the optimum decreases. Thus, for the codes $\BCH(511, 58, 183)$, $\BCH(511, 76, 171)$, and $\EQR(272, 136, 40)$, both algorithms obtain the optimum in all experiments, whilst for 
$\BCH(511, 166, 95)$, $\EQR(432, 216, 48)$ and $\EQR(440, 220, 48)$, the number of experiments that provide the optimum decreases substantially, specially for CHC-Order for the latter code. 
Nevertheless, the proposal still seems to be competitive for length $511$.

On the other hand, if we analyze the worst/average solution with regard to the codes over $\mathbb{F}_8$ in the previous section, we may see that the differences with true distances are greater for the codes studied in this section. This is also expected, since the problem size increases, and the algorithms get trapped into local optima more easily. This also affects the computational time of each experiment, which requires a few minutes to end with the larger codes.

Regarding performance, we analyzed the results of the average minimum distance of the outputs using the Wilcoxon test with \(95\%\) of confidence level. We follow the same methodology that is described in the previous section. In the case of $\mathbb{F}_2$-codes, no significant statistical difference was noticed except for $\EQR(440, 220, 48)$, where GGA-Order is clearly better than CHC-Order. No statistical difference was also found in the computational time for each dataset, except for $\EQR(432, 216, 48)$, where CHC-Order is better. These results are consistent with the analysis carried out over $\mathbb{F}_8$-codes in the previous subsection.

We also compare the results of our approach with previous existing methods in the literature. More specifically, we compare our method over BCH codes with the proposal of Askali et al. in \cite{askali13}, where two variants of genetic algorithms are designed with discrete representation. This work was selected for comparison for the following reasons: (a) there are not much methods in the literature that attempt to calculate the distance of general linear codes, and this is the most recent (2013), as far as we know; (b) it also deals with genetic algorithms, so it is close to our approach; and (c) it calculates the distance for BCH codes of medium size, as we do. The second baseline is the work of Leon \cite{Leon1988}, where a probabilistic method is designed to estimate bounds of some EQR codes. This baseline helps to ensure that our method was able to obtain the true distance for these EQR codes.

Other approaches for specific codes exist in the literature, as for instance for BCH codes \cite{AugotCharpinSendrier1992}, but they impose hard constraints to exploit the code geometry, and therefore they are not comparable with our work, which is a general proposal for any linear code. 

\begin{table}[!ht]
	\centering
	\begin{tabular}{|c||c||c|}
		\hline
		 \textbf{Dataset} & \textbf{CHC-Order} & \textbf{Baselines (\cite{askali13, Leon1988})}\\
		\hline
\(\BCH(511,103,123)\) & \(127\) & \(160\) \cite{askali13}\\
\(\BCH(511,121,117)\) & \(120\) & \(155\) \cite{askali13}\\
\(\BCH(511,166,95)\) & \(95\) & \(135\) \cite{askali13}\\
\(\BCH(511,76,171)\) & \(171\) & \(176\) \cite{askali13}\\
\(\BCH(511,58,183)\) & \(183\) & \(183\) \cite{askali13}\\
\(\EQR(272,136,40)\) & \(40\) & \(40\) ($p=10^{-100}$)\cite{Leon1988}\\
\(\EQR(338,169,40)\) & \(40\) & \(40\) ($p=10^{-100}$)\cite{Leon1988}\\
\(\EQR(368,184,48)\) & \(48\) & \(48\) ($p=10^{-20}$)\cite{Leon1988}\\
\(\EQR(432,216,48)\) & \(48\) & \(48\) ($p=10^{-10}$)\cite{Leon1988}\\
\(\EQR(440,220,48)\) & \(48\) & \(48\) ($p=10^{-10}$)\cite{Leon1988}\\
         \hline
	\end{tabular}
	\caption{Distances found by the proposal CHC-order and the baseline methods.}
	\label{table:askaliComparison}
\end{table}

Table \ref{table:askaliComparison} summarizes the comparison between the results of the CHC-Order and the baseline methods in \cite{askali13, Leon1988}. 
In the case of \cite{Leon1988}, the probability of being the minimum distance is also provided. As we may observe, the proposal described in this manuscript outperforms clearly the baseline methods for BCH codes, and confirms the distance calculated with a given probability in \cite{Leon1988}. In the case of \cite{askali13}, the same conclusion that we pointed out in Section \ref{exp:classicRepr} for discrete representation applies: As the baseline method uses discrete representation, a high selective pressure over the trivial solution $(0, \dots, 0) \in \mathbb{F}_q^k$ makes the algorithms get trapped into local optima. 
The order representation does not suffer from premature convergence, and the search is speeded up by the computation of the RREF, so it improves the results.

\subsection{Discussion}\label{exp:discussion}

Finding the minimum distance of a linear code is an NP-hard problem. Actually, the associated decision problem is NP-complete for binary codes, see \cite{Vardy1997}. In this setting, our experiments have shown that genetic algorithms are good candidates to handle this problem. The results suggest that the discrete representation, in spite of being efficient and adequate to include the classic standard techniques, specially for binary codes,  suffers of premature convergence. Such an issue seems to be due
to the high selective pressure over the zero tuple, which is not a valid solution. Nevertheless, the weight of a message is not necessary correlated with the weight of its encoded word, so that local optima near the zero message could be far from an optimal codeword. Our proposal helps to overcome these drawbacks. This order representation is able to provide good results in conjunction with the strategy of computing the RREF of a generating matrix. Although computing the RREF, as fitness evaluation function, increases the theoretical time complexity from $\mathcal{O}(kn)$ to $\mathcal{O}(k^2n)$, we have shown that this disadvantage 
is palliated by the speed of the convergence to the optimal solution, which requires much less evaluations.  The experiments over the field $\mathbb{F}_8$ reveal significant differences between the standard use of the discrete representation and the algorithms with order representation that we present in this paper. Actually, it is foreseeable achieving bigger differences when working over larger fields. The experiments over $\mathbb{F}_2$ also show that GGA-Order and CHC-Order outperform the best solutions found in \cite{askali13}, and they are be able to find the approximations in \cite{Leon1988}.

Finally, the experiments carried out in Sections \ref{exp:classicRepr} and \ref{exp:previousAppr}  show that there is no significant difference between GGA-Order and CHC-Order regarding accuracy and computational time for the codes of largest length. Actually, they were able to find the true distance for all cases but for the code \(\BCH(511, 121, 117)\), for which we found an upper bound of \(120\) whilst the true distance \cite[Table I]{AugotCharpinSendrier1992} is \(119\).
In this case, we could assess the benefits of both algorithms qualitatively: While the designed GGA-Order requires to set a mutation operator and a population reinitialization criterion, CHC-Order does not use mutation, and it includes internal mechanisms to maintain diversity in the population. Thus, from the point of view of technical experimentation, CHC-Order has the advantage that it requires less parameter settings than GGA-Order and it could be preferable. However, the distance between solutions has to be computed, which makes the algorithm design more complex.

\section{Analysis of the code $(75, 45, 17-27)$ over $\mathbb{F}_8$}\label{sec:magma}

According to \cite{codetables} and \textsc{Magma} \cite{Magma}, the best known linear code over \(\mathbb{F}_8\) with length \(75\) and dimension \(45\) is \((75,45,17-27)\) as it can be verified executing the first 7 lines of the \textsc{Magma} source code in Figure \ref{alg:75_45_17}. However, our algorithms obtain the column permutation in \eqref{eq:754517perm}, which provides a matrix whose RREF contains a row whose weight is \(15\). A codeword with this weight is exhibited in \eqref{eq:754517word}. \textsc{Magma} source code to verify our findings is set in lines 9 to 13 in Figure \ref{alg:75_45_17}, which can be tested in the online version of \textsc{Magma} at the URL \texttt{http://magma.maths.usyd.edu.au/calc/}. Thus, the minimum distance is less or equal to \(15\), which is below of the lower bound provided by \textsc{Magma}. 

\begin{figure}
	\caption{\textsc{Magma} source code exhibiting a codeword with weight \(15\) for the code \((75, 45, 17-27)\)}
	\label{alg:75_45_17}
\begin{tabular}{|c|}
\hline 
\begin{minipage}{0.95\textwidth}
\begin{verbatim}

/* Build the code */
F<a> := GF(8);
V := VectorSpace(F,75);
C := BKLC(F,75,45);
C:Minimal;
/* Show the lower and upper bounds for the code */
BKLCLowerBound(F,75,45), BKLCUpperBound(F,75,45);
/* Set the word of weight 15 */
c := V ! [a^2 + a, 0, 0, 0, 0, 0, 0, 0, 0, 0, 0, 0, 0, 0, a, 0, 0, 1, 0, 0, 0, 0, 0, 
0, 0, 0, 0, 0, 0, 0, a, 0, 0, 0, 0, 0, 0, a^2 + a, a^2, 0, a^2 + a + 1, 1, 0, a + 1, 
0, 0, a, 0, 0, 0, 0, 0, a^2 + a + 1, 0, 0, 0, a + 1, 1, 0, 0, 0, 0, 0, a^2 + a + 1, 
0, 0, 0, 0, 0, 0, a + 1, 0, 0, 0, 0];
/* Checks if the word c belongs to code C */
c in C;
/* Prints the weight of the codeword c in code C */
Weight(c);

\end{verbatim} 
\end{minipage} \\
\hline
\end{tabular}
\end{figure}

\begin{multline}\label{eq:754517perm}
(53,11,58,36,17,27,44,8,73,24,20,71,69,46,10,43,26,61,29,57,23,6,5,67,14,4,50,45,72,59,18,25,\\
47,28,51,68,22,48,52,42,60,21,38,64,16,34,2,3,31,33,49,63,74,54,62,19,7,1, 0,15,13,66,39,65,\\
41,30,12,37,32,70,56,40,55,9,35)
\end{multline}

\begin{multline}\label{eq:754517word}
(a^2 + a, 0, 0, 0, 0, 0, 0, 0, 0, 0, 0, 0, 0, 0, a, 0, 0, 1, 0, 0, 0, 0, 0, 0, 0, 0, 0, 0, 0, 0, a, 0, 0, 0, 0, 0, 0, a^2 + a, a^2, 0,\\
 a^2 + a + 1, 1, 0, a + 1, 0, 0, a, 0, 0, 0, 0, 0, a^2 + a + 1, 0, 0, 0,
 a + 1, 1, 0, 0, 0, 0, 0, a^2 + a + 1, 0, 0, 0, 0, 0, 0,\\
 a + 1, 0, 0, 0, 0)
\end{multline}

According to \cite{codetables}, the linear code \((75, 45, 17-27)\) is constructed from the code with parameters $(77, 46, 18)$ described in \cite{XingLing2000}. By additional experiments we checked if the minimum distance of the code $(77, 46, 18)$ is \(18\). We found out that the column permutation in \eqref{eq:774618perm} leaded to find the codeword shown in \eqref{eq:774618word}, as a row of the RREF of the modified matrix. The weight of this word is also \(15\), which is a new upper bound for the code. Figure \ref{alg:77_46_18} prints the \textsc{Magma} code that ensured these findings.

\begin{multline}\label{eq:774618perm}
(62,4,38,66,31,76,36,56,35,15,57,23,39,18,68,12,40,2,61,32,8,7,21,55,34,65,24,13,\\
25,29,11,46,10,50,9,16,51,19,72,75,71,59,43,44,17,0,26,1,22,74,73,45,3,28,67,63,49,\\
30,33,70,54,42,60,64,53,69,6,14,52,5,20,48,37,41,47,58,27)
\end{multline}

\begin{multline}\label{eq:774618word}
(1, 0, 0, 0, 0, 0, 0, 0, 0, 0, 0, 0, 0, 0, 0, 0, 0, 0, 0, 0, a^2 + 1, 0, a + 1, 0, 0, 0, 0, 1, a + 1, 0, 0, 0, 0, 0, 0, 0, 0, 1, 0,\\ 
0, 0, a + 1, 1, 0, 0, a + 1, 0, 1, 1, 0, 0, 0, 0, 0, 1, 0, 0, 0, 0, 0, 0, 0, 0, 0, a + 1, 0,
  0, 0, 0, 0, a, 0, 0, 0, a^2 + 1, 0, 0)
\end{multline}

\begin{figure}
	\caption{\textsc{Magma} source code exhibiting a codeword with weight \(15\) for the code \((77, 46, 18)\)}
	\label{alg:77_46_18}
\begin{tabular}{|c|}
\hline 
\begin{minipage}{0.95\textwidth}
\begin{verbatim}

/* Build the code */
F<a> := GF(8);
V := VectorSpace(F,77);
C := BKLC(F,77,46);
C:Minimal;
/* Show the lower and upper bounds for the code */
BKLCLowerBound(F,77,46), BKLCUpperBound(F,77,46);
/* Set the word of weight 15 */
c := V ! [1, 0, 0, 0, 0, 0, 0, 0, 0, 0, 0, 0, 0, 0, 0, 0, 0, 0, 0, 0, a^2 + 1, 0, 
a + 1, 0, 0, 0, 0, 1, a + 1, 0, 0, 0, 0, 0, 0, 0, 0, 1, 0, 0, 0, a + 1, 1, 0, 0, 
a + 1, 0, 1, 1, 0, 0, 0, 0, 0, 1, 0, 0, 0, 0, 0, 0, 0, 0, 0, a + 1, 0, 0, 0, 0, 0, 
a, 0, 0, 0, a^2 + 1, 0, 0]; 
/* Checks if the word c belongs to code C */
c in C;
/* Prints the weight of the codeword c in code C */
Weight(c);

\end{verbatim}
\end{minipage} \\
\hline
\end{tabular}
\end{figure}

Additional tests were performed for other codes considered in \cite{XingLing2000}. More specifically, we calculated the permutations in \eqref{eq:754418perm} and \eqref{eq:764518perm} for the codes $(75, 44, 18)$ and $(76, 45, 18)$, respectively, which helps us finding the corresponding codewords in \eqref{eq:754418word} and \eqref{eq:764518word}, both with weight \(15\). The \textsc{Magma} source codes which check that an upper bound for these codes is \(15\) is shown in Figures \ref{alg:75_44_18} and \ref{alg:76_45_18}, respectively.

\begin{multline}\label{eq:754418perm}
(33,23,13,29,53,30,71,7,12,35,34,57,60,39,52,5,24,4,56,44,31,25,43,40,11,6,\\
22,32,26,63,62,17,18,38,9,72,2,8,19,69,36,73,21,65,42,68,14,10,20,59,58,27,28,\\
51,66,47,48,37,67,54,70,74,16,46,3,1,15,50,55,0,41,45,61,64,49)
\end{multline}

\begin{multline}\label{eq:754418word}
(a^2 + a, 0, 0, 0, 0, 0, 0, 0, 0, 0, 0, 0, 0, 0, 0, 0, 0, 0, 0, 0, a + 1, 0, 1, 0, 0, 0, 0, a^2 + a, 1,\\
0, 0, 0, 0, 0, 0, 0, 0, a^2 + a, 0, 0, 0, 1, a^2 + a, 0, 0, 1, 0, a^2 + a, a^2 + a, 0, 0, 0, 0, 0, a^2 + a, \\
0, 0, 0, 0, 0, 0, 0, 0, 0, 1, 0, 0, 0, 0, 0, a^2 + a + 1, 0, 0, 0, a + 1)
\end{multline}

\begin{multline}\label{eq:764518perm}
(25,56,43,14,58,11,61,57,29,10,26,51,31,8,32,73,7,53,46,60,4,12,35,67,59,2,55,\\
36,17,6,34,24,16,72,52,44,21,9,71,19,65,62,68,45,49,75,15,42,0,48,64,13,28,40,\\
23,27,33,3,70,37,1,50,74,54,38,41,30,63,69,5,22,18,20,66,39,47)
\end{multline}

\begin{multline}\label{eq:764518word}
(a^2 + a, 0, 0, 0, 0, 0, 0, 0, 0, 0, 0, 0, 0, 0, 0, 0, 0, 0, 0, 0, a + 1, 0, 1, 0, 0, 0, 0, a^2 + a, 1, 0, 0, 0, 0, 0, 0, 0, 0,\\
 a^2 + a, 0, 0, 0, 1, a^2 + a, 0, 0, 1, 0, a^2 + a, a^2 + a, 0, 0, 0, 0, 0, a^2 + a, 0, 0, 0, 0, 0, 0, 0, 0, 0, 1, 0, 0, 0, 0, 0, \\
 a^2 + a + 1, 0, 0, 0, a + 1, 0)
\end{multline}

\begin{figure}
	\caption{\textsc{Magma} source code exhibiting a codeword with weight \(15\) for the code \((75, 44, 18)\)}
	\label{alg:75_44_18}
\begin{tabular}{|c|}
\hline 
\begin{minipage}{0.95\textwidth}
\begin{verbatim}

/* Build the code */
F<a> := GF(8);
V := VectorSpace(F,75);
C := BKLC(F,75,44);
C:Minimal;
/* Show the lower and upper bounds for the code */
BKLCLowerBound(F,75,44), BKLCUpperBound(F,75,44);
/* Set the word of weight 15 */
c := V ! [a^2 + a, 0, 0, 0, 0, 0, 0, 0, 0, 0, 0, 0, 0, 0, 0, 0, 0, 0, 0, 0, a + 1, 
0, 1, 0, 0, 0, 0, a^2 + a, 1, 0, 0, 0, 0, 0, 0, 0, 0, a^2 + a, 0, 0, 0, 1, a^2 + a, 
0, 0, 1, 0, a^2 + a, a^2 + a, 0, 0, 0, 0, 0, a^2 + a, 0, 0, 0, 0, 0, 0, 0, 0, 0, 1, 
0, 0, 0, 0, 0, a^2 + a + 1, 0, 0, 0, a + 1];
/* Checks if the word c belongs to code C */
c in C;
/* Prints the weight of the codeword c in code C */
Weight(c);

\end{verbatim}
\end{minipage} \\
\hline
\end{tabular}
\end{figure}

\begin{figure}
	\caption{\textsc{Magma} source code exhibiting a codeword with weight \(15\) for the code \((76, 45, 18)\)}
	\label{alg:76_45_18}
\begin{tabular}{|c|}
\hline 
\begin{minipage}{0.95\textwidth}
\begin{verbatim}

/* Build the code */
F<a> := GF(8);
V := VectorSpace(F,76);
C := BKLC(F,76,45);
C:Minimal;
/* Show the lower and upper bounds for the code */
BKLCLowerBound(F,76,45), BKLCUpperBound(F,76,45);
/* Set the word of weight 15 */
c := V ! [a^2 + a, 0, 0, 0, 0, 0, 0, 0, 0, 0, 0, 0, 0, 0, 0, 0, 0, 0, 0, 0, a + 1, 
0, 1, 0, 0, 0, 0, a^2 + a, 1, 0, 0, 0, 0, 0, 0, 0, 0, a^2 + a, 0, 0, 0, 1, a^2 + a, 
0, 0, 1, 0, a^2 + a, a^2 + a, 0, 0, 0, 0, 0, a^2 + a, 0, 0, 0, 0, 0, 0, 0, 0, 0, 1, 
0, 0, 0, 0, 0, a^2 + a + 1, 0, 0, 0, a + 1];
/* Checks if the word c belongs to code C */
c in C;
/* Prints the weight of the codeword c in code C */
Weight(c);

\end{verbatim}
\end{minipage} \\
\hline
\end{tabular}
\end{figure}

\section{Conclusions}\label{Section4}

We have shown that, for a given generating matrix of a linear code, there exists a column permutation for which its reduced row echelon form has a row whose weight reaches the minimum distance. This opens the possibility to use a permutation representation in metaheuristics to find the minimum distance of an arbitrary linear code, as an alternative to the classic discrete representation. The proposed model is polynomial time-dependent with respect to the dimension of the code and the size of the base finite field.
We have developed different schemes of evolutionary algorithms to prove the benefits of our approach experimentally. We have concluded that usual limitations of discrete representation, such as high selective pressure, are palliated by means of the new representation. Therefore, our approach is able to find true minimum distances of general linear codes of medium size, and it is the first work to address the problem for codes over finite fields larger than $\mathbb{F}_2$ using metaheuristics. Comparison with the existing methods in the literature suggests that our proposal is accurate and efficient for BCH and EQR codes over $\mathbb{F}_2$, and able to outperform previous approaches. In addition, we have been able to find some inacuracies in the list of best known linear codes in \cite{codetables} and \cite{Magma}.

The run time of the calculation of the reduced row echelon form increases polynomically with the code length and dimension. Future works will be conducted to reduce this execution time, and also to include information about specific code geometry and properties into the evolutionary process to speed up convergence to optimal solutions.

\section*{Acknowledgements}

Research partially supported by grant \textup{PID2019-110525GB-I00}
from {Agencia Estatal de Investigaci\'{o}n (AEI) and from Fondo Europeo de Desarrollo Regional (FEDER)}, and by grant \textup{A-FQM-470-UGR18} from Consejer\'{\i}a de Econom\'{\i}a y Conocimiento de la Junta de Andaluc\'{\i}a and Programa Operativo FEDER 2014-2020. 

\bibliographystyle{elsarticle-num}

\begin{thebibliography}{10}
\expandafter\ifx\csname url\endcsname\relax
  \def\url#1{\texttt{#1}}\fi
\expandafter\ifx\csname urlprefix\endcsname\relax\def\urlprefix{URL }\fi
\expandafter\ifx\csname href\endcsname\relax
  \def\href#1#2{#2} \def\path#1{#1}\fi

\bibitem{Shannon1948}
C.~E. Shannon, A mathematical theory of communication, Bell System Technical
  Journal 27 (1948) 379--423 \& 623--656.

\bibitem{Huffmann/Pless:2003}
W.~C. Huffmann, V.~Pless, Fundamentals of Error-Correcting Codes, Cambridge
  University Press, 2003.

\bibitem{Vardy1997}
A.~{Vardy}, The intractability of computing the minimum distance of a code,
  IEEE Transactions on Information Theory 43~(6) (1997) 1757--1766.

\bibitem{Shor94}
P.~W. Shor, Algorithms for quantum computation: Discrete logarithms and
  factoring, in: Proceedings of the 35th Annual Symposium on Foundations of
  Computer Science, IEEE Computer Society Press, Los Alamitos, CA, 1994, pp.
  124 -- 134.

\bibitem{Nis16}
L.~Chen et~al., Report on post-quantum cryptography, Tech. Rep. NISTIR 8105,
  National Institute of Standards and Technology. (2016).

\bibitem{PQCrypto}
E.~D. D.~J.~Bernstein, J.~Buchmann, Post-Quantum Cryptography, Springer-Verlag
  Berlin Heidelberg, 2009.

\bibitem{Berlekamp2015}
E.~R. Berlekamp,
 {Algebraic Coding
  Theory}, World Scientific, 2015.

\bibitem{Wicker1994}
S.~B. Wicker, Reed-Solomon Codes and Their Applications, Wiley-IEEE Press, 1994.

\bibitem{Wassermann/etal:2006}
A.~Betten, M.~Braun, H.~Fripertinger, A.~Kerber, A.~Kohnert, A.~Wassermann,
  Error-Correcting Linear Codes, Vol.~18 of Algorithms and Computation in
  Mathematics, Springer, 2006.

\bibitem{LisonekTrummer2016}
P.~Lisonek, L.~Trummer, Algorithms for the minimum weight of linear codes, Adv.
  in Math. of Comm. 10 (2016) 195--207.


\bibitem{joundan19}
I.~Joundan, S.~Nouh, M.~Azouazi, A.~Namir, A new efficient way based on special
  stabilizer multiplier permutations to attack the hardness of the minimum
  weight search problem for large bch codes, International Journal of
  Electrical and Computer Engineering (IJECE) 9 (2019) 1232 -- 1239.

\bibitem{Leon1988}
J.~S. {Leon}, A probabilistic algorithm for computing minimum weights of large
  error-correcting codes, IEEE Transactions on Information Theory 34~(5) (1988)
  1354--1359.

\bibitem{DontasDeJong1990}
K.~{Dontas}, K.~{De Jong}, Discovery of maximal distance codes using genetic
  algorithms, in: [1990] Proceedings of the 2nd International IEEE Conference
  on Tools for Artificial Intelligence, 1990, pp. 805--811.


\bibitem{muxiang_sa_94}
Z.~Muxiang, M.~Fulong, Simulated annealing approach to the minimum distance of
  error-correcting codes, International Journal of Electronics 76 (1994), 377--384.

\bibitem{BLAND2007391}
J.~Bland,
 {Local
  search optimisation applied to the minimum distance problem}, Advanced
  Engineering Informatics 21~(4) (2007) 391 -- 397.


\bibitem{8360246}
H.~{Bouzkraoui}, A.~{Azouaoui}, Y.~{Hadi}, New ant colony optimization for
  searching the minimum distance for linear codes, in: 2018 International
  Conference on Advanced Communication Technologies and Networking (CommNet),
  2018, pp. 1--6.

\bibitem{10.1007/978-3-540-73055-2_51}
J.~E. Amaya, C.~Cotta, A.~J. Fern{\'a}ndez, Tackling the error correcting code
  problem via the cooperation of local-search-based agents, in: J.~Mira, J.~R.
  {\'A}lvarez (Eds.), Nature Inspired Problem-Solving Methods in Knowledge
  Engineering, Springer Berlin Heidelberg, Berlin, Heidelberg, 2007, pp.
  490--500.

\bibitem{BIA0002332}
M.~B.~B. Aylaj, M.~Belkasmi, New simulated annealing algorithm for computing
  the minimum distance of linear block codes, Advances in Computational
  Research 6~(1) (2014) 153 --158.

\bibitem{BarbieriCagnoniColavolpe2005}
A.~Barbieri, S.~Cagnoni, G.~Colavolpe, Genetic design of linear block
  error-correcting codes, in: B.~Apolloni, M.~Marinaro, R.~Tagliaferri (Eds.),
  Biological and Artificial Intelligence Environments, Springer Netherlands,
  Dordrecht, 2005, pp. 107--116.

\bibitem{Nouhetal13}
S.~Nouh, M.~Belkasmi, Genetic algorithms for finding the weight enumerator of
  binary linear block codes, International Journal of Applied Research on
  Information Technology and Computing 2 (3) 2011, 80 -- 93.


\bibitem{7381308}
A.~{Berkani}, A.~{Azouaoui}, M.~{Belkasmi}, Soft-decision decoding by a compact
  genetic algorithm using higher selection pressure, in: 2015 International
  Conference on Wireless Networks and Mobile Communications (WINCOM), 2015, pp.
  1--6.

\bibitem{Berkani2017ImprovedDO}
A.~Berkani, A.~Azouaoui, M.~Belkasmi, B.~Aylaj, M.~V.-R. University, S.~labo,
  Improved decoding of linear block codes using compact genetic algorithms with
  larger tournament size, International Journal of Computer Science Issues
  14~(1) 2017, 15 -- 24.


\bibitem{askali13}
M.~Askali, A.~Azouaoui, S.~Nouh, M.~Belkasmi, On the computing of the minimum
  distance of linear block codes by heuristic methods, International Journal of
  Communications, Network and System Sciences 5 (11) 2012, 774 -- 784.


\bibitem{joundan18}
I.~Joundan, S.~Nouh, A.~Namir, A new way to find the minimum distance of bch
  and quadratic double circulant codes, SSRN Electronic Journal, 2018,\href
  {http://dx.doi.org/10.2139/ssrn.3176115} {\path{doi:10.2139/ssrn.3176115}}.


\bibitem{Tomas/Rosenthal/Smarandache:2010}
V.~Tom{\'a}s, J.~Rosenthal, R.~Smarandache, Decoding of convolucional codes
  over the erasure channel., {IEEE} Transactions on Information Theory 58~(1)
  (2010), 90--108.

\bibitem{Lidl/Niederreiter:1996}
R.~Lidl, H.~Niederreiter, Finite Fields, 2nd Edition, Encyclopedia of
  Mathematics and its Applications, Cambridge University Press, 1996.

\bibitem{Liesen/Mehrmann:2015}
J.~Liesen, V.~Mehrmann, Linear Algebra, Springer Undergraduate Mathematical
  Series, Springer International Publishing, 2015.

\bibitem{Bose/Chaudhuri:1960}
R.~C. Bose, D.~K. Ray-Chaudhuri, On a class of error correcting binary group
  codes, Information and Control 3~(1) (1960), 68 -- 79.


\bibitem{Hocquenghem:1959}
A.~Hocquenghem, Codes correcteurs d'erreurs, Chiffres 2 (1959), 147--156.

\bibitem{Hartmann/Tzeng:1972}
C.~R.~P. Hartmann, K.~K. Tzeng, Generalizations of the {BCH}-bound.,
  Information and Control 20 (1972), 489--498.

\bibitem{Gomez/Lobillo/Navarro/Neri:2018}
J.~G{\'{o}}mez{-}Torrecillas, F.~J. Lobillo, G.~Navarro, A.~Neri,
  {Hartmann--Tzeng} bound and skew cyclic codes of designed {Hamming}
  distance., Finite Fields and Their Applications 50 (2018), 84--112.


\bibitem{OverbeckSendrier2009}
R.~Overbeck, N.~Sendrier, Code-based cryptography, in: D.~J. Bernstein,
  J.~Buchmann, E.~Dahmen (Eds.), Post-Quantum Cryptography, Springer-Verlag,
  Berlin, Heidelberg, 2009, pp. 95--145.


\bibitem{GLN2019}
J.~G{\'o}mez-Torrecillas, F.~J. Lobillo, G.~Navarro,
{Minimum distance computation
  of linear codes via genetic algorithms with permutation encoding}, ACM
  Communications in Computer Algebra 52~(3) (2019), 71--74.

\bibitem{BunchHopcroft1974}
J.~R. Bunch, J.~E. Hopcroft, Triangular factorization and inversion by fast
  matrix multiplication, Mathematics of Computation 28~(125) (1974), 231--236.

\bibitem{AndrenHellstromMarkstrom2007}
D.~Andr{\'e}n, L.~Hellstr{\"o}m, K.~Markstr{\"o}m, On the complexity of matrix
  reduction over finite fields, Advances in Applied Mathematics 39 (2007),
  428--452.

\bibitem{eshelman1991}
L.~J. Eshelman,
 {The
  {CHC} adaptive search algorithm: How to have safe search when engaging in
  nontraditional genetic recombination}, in: G.~J. RAWLINS (Ed.), Foundations
  of Genetic Algorithms, Vol.~1 of Foundations of Genetic Algorithms, Elsevier,
  1991, pp. 265 -- 283.


\bibitem{eshelman93}
L.~J. Eshelman, J.~D. Schaffer, Real-coded genetic algorithms and
  interval-schemata, in: Foundations of Genetic Algorithms, Vol.~2 of
  Foundations of Genetic Algorithms, Elsevier, 1993, pp. 187 -- 202.

\bibitem{cordon2006}
O.~Cord{\'o}n, S.~Damas, J.~Santamar{\'\i}a, Feature-based image registration
  by means of the chc evolutionary algorithm, Image and Vision Computing 24~(5)
  (2006), 525 -- 533.

\bibitem{chcdiscrete1}
A.~Sim\~{o}es, E.~Costa,
{Memory-based {CHC} algorithms
  for the dynamic traveling salesman problem}, in: Proceedings of the 13th
  Annual Conference on Genetic and Evolutionary Computation, GECCO '11,
  Association for Computing Machinery, New York, NY, USA, 2011, pp. 1037--1044.

\bibitem{BackFogelMichalewicz1997}
T.~Back, D.~B. Fogel, Z.~Michalewicz, Handbook of Evolutionary Computation, 1st
  Edition, IOP Publishing Ltd., GBR, 1997.

\bibitem{Blickle97acomparison}
T.~Blickle, L.~Thiele, A comparison of selection schemes used in evolutionary
  algorithms, Evolutionary Computation 4 (1997), 361--394.

\bibitem{Gwiazda2006}
T.~D. Gwiazda, Genetic Algorithms Reference, Tomasz Gwiazda, 2006.

\bibitem{dixon}
J.~D. Dixon, The probability of generating the symmetric group, Mathematische
  Zeitschrift 110 (1969), 199--205.

\bibitem{Magma}
W.~Bosma, J.~Cannon, C.~Playoust,
  \href{http://dx.doi.org/10.1006/jsco.1996.0125}{The {M}agma algebra system.
  {I}. {T}he user language}, J. Symbolic Comput. 24~(3-4) (1997) 235--265,
  computational algebra and number theory (London, 1993).

\bibitem{codetables}
M.~Grassl, Bounds on the minimum distance of linear codes and quantum codes,
  Online available at \url{http://www.codetables.de}, accessed 2020-01-22
  (2007).

\bibitem{XingLing2000}
C.~Xing, S.~Ling, A class of linear codes with good parameters from algebraic
  curves, IEEE Transactions on Information Theory 46~(4) (2000) 1527--1532.

\bibitem{AugotCharpinSendrier1992}
D.~{Augot}, P.~{Charpin}, N.~{Sendrier}, Studying the locator polynomials of
  minimum weight codewords of {BCH} codes, IEEE Transactions on Information
  Theory 38~(3) (1992) 960--973.

\end{thebibliography}








\end{document}